\documentclass{IEEEtran}  
\usepackage{cite}
\ifCLASSINFOpdf
\else
\fi 
\usepackage{caption}
\usepackage{tikz, pgfplots}
\usepackage{pgfplots}
\usepackage{blindtext}
\usepackage{graphicx} 
\usepackage{mdwmath}
\usepackage{mdwtab}
\usepackage{bm} 
\usepackage{amsmath}
\usepackage{amssymb}
\usepackage{multicol} 
\usepackage{float}
\usepackage{epstopdf}
\usepackage{cancel}
\usepackage{cases} 
\usepackage{subfig}
\usepackage{color}  
\newenvironment{proof}[1][Proof]{\textbf{#1.} }{\ \rule{0.5em}{0.5em}}  
\newtheorem{Proposition}{Proposition}
\newcommand{\snr}{\mathrm{SNR}}
\newcommand{\sinr}{\mathrm{SINR}}  
\newcommand{\hd}{\mathrm{Hf}}
\newcommand{\hdc}{\mathrm{Hf_{Cl}}}
\newcommand{\hdo}{\mathrm{Hf_{Op}}}
\newcommand{\os}{\mathrm{SqBf}}
\newcommand{\osc}{\mathrm{SqBf_{Cl}}}
\newcommand{\oso}{\mathrm{SqBf_{Op}}}
\newcommand{\se}{\mathrm{SE}}
\newcommand{\zf}{\mathrm{ZF}}
\newcommand{\ini}{\mathrm{INI}}
\newcommand{\tr}{\mathrm{tr}}
\newcommand{\data}{\mathrm{data}}     

\begin{document} 
\title{Sequential Beamforming for Multiuser MIMO with Full-duplex Training}
\author{Xu Du, John Tadrous and Ashutosh Sabharwal  
\thanks{Xu Du, Ashutosh Sabharwal are with the Department of Electrical and Computer Engineering, Rice University, Houston, TX, 77005 (e-mails: \{Xu.Du, ashu\}@rice.edu).
 John Tadrous was with with the Department of Electrical and Computer Engineering, Rice University, Houston. He is now with Electrical and Computer Engineering Department, Gonzaga University, Spokane, WA, 99223 (e-mail: tadrous@gonzaga.edu) .
 
 The material in this paper was presented in part at the Asilomar conference on signals systems and computers, 2014 and 16th IEEE International Workshop on Signal Processing Advances in Wireless Communications, 2015. This work was partially supported by a grant from Intel and Verizon Labs, and NSF Grant CNS-1314822.} 
}
\bibliographystyle{IEEEbib}
\maketitle 
\begin{abstract}
Multiple transmitting antennas can considerably increase the downlink spectral efficiency by beamforming to multiple users at the same time. However, multiuser beamforming requires channel state information (CSI) at the transmitter, which leads to training overhead and reduces overall achievable spectral efficiency.  In this paper, we propose and analyze a sequential beamforming strategy that utilizes full-duplex base station to implement downlink data transmission concurrently with CSI acquisition via in-band closed or open loop training.  Our results demonstrate that full-duplex capability can improve the spectral efficiency of uni-directional traffic, by leveraging it to reduce the control overhead of CSI estimation. In moderate $\snr$ regimes, we analytically derive tight approximations for the optimal training duration and characterize the associated respective spectral efficiency. We further characterize the enhanced multiplexing gain performance in the high $\snr$ regime. In both regimes, the performance of the proposed full-duplex strategy is compared to the half-duplex counterpart to quantify spectral efficiency improvement. With experimental data [1] and 3D channel model [2] from 3GPP, in a $1.4 $ MHz $8\times 8$ system LTE system with the block length of $500$ symbols, the proposed strategy attains a spectral efficiency improvement of 130\% and 8\% with closed and open loop training, respectively. 
\end{abstract}
\IEEEpeerreviewmaketitle
\newtheorem{Definition}{Definition}
\newtheorem{theorem}{Theorem}
\newtheorem{Lemma}{Lemma}
 \newtheorem{proposition}{Proposition}
 \newtheorem{remark}{Remark}
  \newtheorem{observation}{Observation}
\section{Introduction} 
Multiuser MIMO downlink systems have the potential to increase the spectral efficiency by serving multiple users at the same time with a multiple-antenna base station. A base station with $M$ antennas can simultaneously support up to $M$ \emph{half-duplex}  single-antenna users at full multiplexing gain, if it has perfect channel information (CSI). Accurate channel knowledge at the transmitter is vital for achieving maximum spectral efficiency. For example, when no CSI is available at the base station, TDMA strategy is optimal~\cite{jafar2005isotropic}. Therefore, only one user can be supported with full multiplexing gain.  In transmitter beamforming based systems, CSI is obtained by either closed or open loop training, which is defined as below. 
\begin{itemize}
\item In \textsl{closed loop training} method, each user first estimates CSI by using the training pilots sent out by the base station. Then, the CSI is quantized and sent back to the base station\footnote{ Closed loop training with analog feedback is possible, but not considered in this paper. In~\cite{caire2010multiuser}, it is shown that, compared to open loop training, with the same amount of training symbols, closed loop training with analog feedback results in a larger error in uplink CSI training and greater interference due to precoding. In this paper, we only consider closed loop training with digital training and open loop training.}.

\item In \textsl{open loop training} method, base station learns downlink CSI by receiving training pilots from users through the uplink channel;  channel \emph{reciprocity} is then leveraged to estimate the downlink CSI from the uplink receptions.  
\end{itemize}
%{  A full-duplex system, where the same antennas are used for transmission and reception, leads to a case that channel reciprocity holds and thus either closed or open loop training can be adopted.} In a half-duplex system, uplink training consumes time resources, which results in less downlink data transmission time.  
For time-varying channels, overhead due to CSI acquisition leads to significant spectral efficiency loss. Thus, CSI estimation overhead reduction remains an important challenge. In this paper, we investigate the use of full-duplex capability to \emph{reduce overhead of CSI estimation} to increase the spectral efficiency of downlink traffic. 

{   The recently developed full-duplex radios~\cite{bharadia2014full,duarte2012experiment,Jain2011practicalFD,everett2015measurement,aryafar2012midu,yin2013full,ahmed2013self} allow concurrent uplink and downlink data transmission. However, the limited receiver dynamic range~\cite{day2012full} and circuit-design~\cite{sabharwal2014band} in small form-factor handsets for full-duplex transmission remain challenging problems. In this paper, we assume that only the base station antenna array to be an $M$-antenna full-duplex MIMO array~\cite{bharadia2014full} and all the mobile nodes to be half-duplex. The potential rate gain region of a full-duplex node is analyzed in~\cite{ahmed2013rate}. In~\cite{bai2013distributed}, a full-duplex base station is used to increase spectral efficiency by serving half-duplex downlink and uplink traffic simultaneously. In this paper, we propose an alternative use of full-duplex, where full-duplex capability is harnessed to increase downlink spectral efficiency by saving on the training time and thus reducing control channel overhead.} 
In particular, our key contributions in this paper are as follows:
%In particular, our key contributions in this paper are as follows {\bf Xu, these bullets don't really read like contributions. I suggest we claim only one contribution -  ``propose and analyze sequential beamforming " Write three para. First, explain the main idea of sequential beamforming. Second, explain how it is analyzed and what the analysis says. Third, how it improves performance for LTE like example. You can use current text. In general, always write contributions in the sequence they appear in the actual manuscript.}
{ 
\begin{itemize}
\item We propose a sequential beamforming strategy for multiuser downlink transmissions with either closed or open loop training. Instead of waiting to receive all CSI before starting data transmission, the base station now begins transmitting to some users as it receives their CSI. %In this scheme, only the base station has to be full-duplex { while \emph{all mobiles are half-duplex}. }

\item The simultaneous transmission of feedback and data creates additional inter-node interference at the downlink receiving users. By optimizing the optimal training duration, we then analyze the spectral efficiency of the proposed sequential beamforming strategy and demonstrate its relative spectral efficiency gain over the half-duplex counterpart in closed-form. In general, due to the relatively low user training power\footnote{In current systems, the transmission power of users is usually limited due to lower power budget compared to a base station. }, we show that inter-node interference only leads to a limited downlink rate reduction during training.

\item The spectral efficiency predicted by the closed-form results are further verified through simulations based on channel data from experiment~\cite{everett2015measurement} and 3GPP 3D channel model~\cite{3gpp.36.873}. For example, in a typical $1.4$ MHz LTE system with a block length of around $500$ symbols, the proposed strategy demonstrates a spectral efficiency improvement of 130\% and 8\% over its half-duplex counterpart, for an $8 \times 8$ multiuser MIMO system with closed and open loop training, respectively.
\end{itemize}}
The rate loss due to imperfect CSI with different types of training has been studied in~\cite{caire2010multiuser}. In~\cite{kobayashi2011training}, the authors characterize the optimal training duration and its associated spectral efficiency for half-duplex systems. User selection~\cite{yoo2007multi,dimic2005downlink} has been proposed to reduce the number of training symbols needed by selecting users with a larger distance in channel space. Our analysis has two main differences from prior research. First, we study how to utilize the full-duplex operation to obtain gains in spectral efficiency. Second, the influence of limited training power at the mobile user is modeled and examined throughout this paper. 
 
We first proposed to utilize the training time in systems composed of both full-duplex base station and \textsl{full-duplex} mobile in~\cite{du2014mimo} and~\cite{du2015SPAWC}. { In this paper, we consider a system comprising a full-duplex base station and only \textsl{half-duplex} users.  }

The remainder of this paper is structured as follows. Section~\ref{sec:SysModel} describes the system model. Then the \textsl{sequential beamforming} strategy is proposed in Section~\ref{sec:PrbFoM}. The optimal training duration is studied in Section~\ref{sec:optLen} for systems with both closed and open loop training. The associated spectral efficiency is then presented in Section~\ref{sec:optSE} with both theoretical analysis and experimental data validation. High $\snr$ analysis is provided in Section~\ref{sec:HiSNRAna} to evaluate the proposed strategy. We conclude this paper by summarizing the main results in Section~\ref{sec:conclude}. 
\section{System Model}\label{sec:SysModel}

{  We consider a symmetric multiuser MIMO downlink system consisting of an $M$-antenna full-duplex base station and $M$ single-antenna non-cooperative half-duplex users.} The base station aims at delivering downlink data to each user.  Albeit sub-optimal, base station adopts zero-forcing (ZF) beamforming~\cite{spencer2004zero} for simultaneous transmission to multiple users. In ZF, the base station projects the signal intended for one user to the null space of the others. Thus, if perfect CSI is available, each user only receives the expected signal without interference.
 
Since CSI is obtained from finite training, it is almost always inaccurate and thus results in inter-beam interference for ZF transmissions. During the full-duplex training, the downlink data is communicated in the same band as the training signals sent by users, thus receiving users also suffer from \textsl{inter-node interference}. In this paper, we quantify the impact of inter-node interference on spectral efficiency for training-based ZF strategy.  
\begin{figure}[htbp]
\centering
\includegraphics[height=3cm]{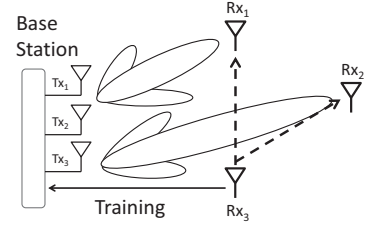} 
\caption{ A schematic of the interference in a $3 \times 3$ multiuser MIMO downlink system when User $3$ sends training and others receive downlink data. The receiving users suffer inter-beam interference (side lobes) due to imperfect CSI. The receiving users also incur inter-node interference (dashed lines) resulting from User $3$'s training. Since users are half-duplex nodes, User $3$ does not receive while it is sending the training signal. }
\label{fig:interference}
\end{figure}  
When User $k$ sends training symbols and the base station transmits downlink data to User $1,2,...,k-1$, the received signal of User $i$ is immediately captured as{ 
\begin{equation}\label{eq:SysModel} 
y_{i}=\textbf{h}_{i}^{*}\textbf{V} \textbf{s}+h_{\mathrm{U_{ik}}}x_{\mathrm{tr}_{k}}+n_{i}  ,\quad i=1,\cdots,k-1.
\end{equation}
 Here $\textbf{h}_{i} \in \mathcal{C}^{1\times M}$ and $h_{\mathrm{U_{ik}}}$ stand for the channel realization from User $i$ to the base station and User $k$, respectively.} {  In this paper, the coherence time length is $T^{\mathrm{coherence}}$ symbols where the channel stays unchanged. Moreover, the block length $T$ is the number of symbols for uplink training and downlink data transmission.} We assume a Rayleigh block fading environment, i.e., each element of $\textbf{h}_{i}$ and $h_{\mathrm{U_{ik}}}$ is independently complex Gaussian distributed from block to block. 
 
The term $\textbf{s}\in \mathcal{C}^{k-1 \times 1}$ is the actual signal intended to User $1,2,\cdots,k-1$ and $\textbf{V}=[\textbf{v}_1, ..., \textbf{v}_{k-1}] \in \mathcal{C}^{M \times k-1}$ represents the ZF precoding matrix generated based on the quantized (estimated) CSI of users, which is presented as $\hat{\mathbf{h}}_i, i=1,2,..,k-1$. The precoded symbol is then $\textbf{V} \textbf{s}$, which is constrained to an average power constraint of $P$. We consider equal power allocation among symbols, i.e., $\mathbb{E}\left[|\textbf{v}_i s_i|^2\right]=P/M, \forall i$. % each of the downlink symbols has power $P/M$, which is mathematically captured as

If only imperfect CSI is available, the inter-beam interference is non-zero. The signal and the inter-beam interference both are contained in term $\textbf{h}_{i}^{*}\textbf{V} \textbf{s}$. Term $x_{\mathrm{tr}_{k}}$ is the uplink training symbol sent by User $k$. To account for the limitations of both battery and size of user devices, we consider a more strict average power constraint for users, which is described as $\mathbb{E}[|x_{\mathrm{tr}_{k}}|^2]\leqslant fP,\ f\in (0,1]$. Term $h_{\mathrm{U_{ik}}}x_{\mathrm{tr}_{k}}$ captures the inter-node interference. { We assume inter-node interference power to be proportional to the training power $fP$, i.e., it grows as $|h_{\mathrm{U_{ik}}}x_{\mathrm{tr}_{k}}|^2\sim\alpha fP$, where $\alpha=\frac{|h_{\mathrm{U_{ik}}}x_{\mathrm{tr}_{k}}|^2}{|x_{\mathrm{tr}_{k}}|^2}>0$ and $fP=|x_{\mathrm{tr}_{k}}|^2>0$.} The signal is degraded by an independent unit variance additive white complex Gaussian noise $n_{i}$.

We assume that each User $i$ has perfect knowledge of the downlink channel from the base station to itself $\textbf{h}_{i}$ by estimating downlink pilots broadcast by base station antennas before any uplink training or downlink data transmission. {  However, the base station is required to obtain CSI by either closed or open loop training. The base station is assumed to be an $M-$antenna full-duplex MIMO array~\cite{bharadia2014full}. We assume that self-interference due to the full-duplex operation at the base station is reduced to near-noise floor by the state-of-art circulators and filter-based self-interference cancellation~\cite{bharadia2014full}. Since the same antennas are used for both transmission and reception, channel reciprocity also holds between uplink and downlink channels. Unlike~\cite{du2014mimo,du2015SPAWC}, the proposed sequential beamforming strategy needs only half-duplex users.} Thus there is no self-interference at mobile nodes.  
\section{Sequential Beamforming} \label{sec:PrbFoM} 
In this section, we propose a sequential beamforming strategy that leverages the full-duplex capability at the base station to send downlink data during CSI collection. First, we describe the sequential beamforming strategy in Section~\ref{sec:ABprop}. Then we characterize the influence of inter-node and inter-beam interference on downlink rate.
\subsection{Sequential Beamforming Strategy}\label{sec:ABprop}
{  The proposed strategy is referred to as \textsl{sequential beamforming}. In sequential beamforming, pre-scheduled users send their channel state information in orthogonal time slots. And as the base station receives a particular user's information, it starts data transmission to that user.} Thus, unlike the half-duplex system, the base station does not wait for all the users to send their channel feedback. As noted before, the proposed strategy only requires the base station to be full-duplex and all the mobiles can be half-duplex.  A sequential beamforming strategy with total $T^{\tr}$ training symbols from all users is described as follows:
\begin{enumerate}
\item At the beginning of each block, no downlink data transmission is performed due to the lack of CSI knowledge. From Symbol $1$ to Symbol $\frac{T^{\tr}}{M}$, User $1$ sends\footnote{For a given set of users, the user index $1,2,..., M$ may be randomly assigned in every coherence block to achieve fairness among users.} training symbols to the base station.  We define symbols from Symbol $\left(j-1\right)\frac{T^{\tr}}{M}+1$ to Symbol $j\frac{T^{\tr}}{M}$ to be Cycle $j$ where User $j$ sends its training symbols. 
\item In cycle $j+1$, the base station transmits downlink data based on the updated ZF precoding matrix and beamforms to User $1,2,.., j$ whose training symbols are collected over the previous $j$ cycles. Users who have finished training, i.e., User $1,2, \cdots, j$, begin receiving downlink data. All receiving users decode the received signal by treating interference (both inter-beam and inter-node interference) as noise.
\item Repeat 2) till the end of $T^{\tr}$ symbols. The above full-duplex training part is referred to as \textsl{training phase}.
\item After all training is collected, only downlink data transmission takes place. This part is referred to as \textsl{half-duplex phase}. Fig.~\ref{fig:AB} provides an illustration of sequential beamforming. 
\end{enumerate}
%\vspace*{-0.5\baselineskip}
\begin{figure*}[htbp]
\centering
\includegraphics[height=3cm]{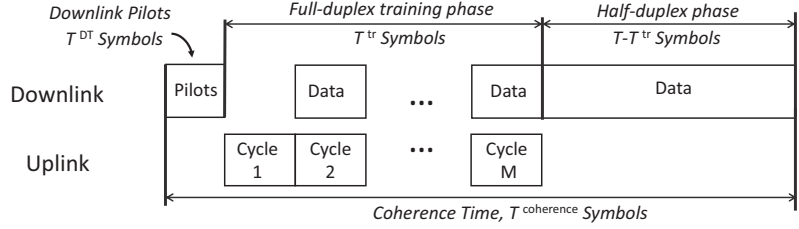}
\caption{{   A depiction of sequential beamforming: Users first obtain CSI by estimating downlink training pilots broadcast by base station antennas in the first $T^{DT}$ symbols. Training cycles constitute $T^{\tr}$ symbols and half-duplex phase occupies the rest of $T-T^{\tr}$ symbols. At the end of training cycles, the base station updates its precoding vectors and serves all users whose CSI has been collected. The whole coherence time length is $T^{\mathrm{coherence}}$ symbols and parameter $T$ is the total number of symbols for uplink training and downlink data transmission.}}
\label{fig:AB}
\end{figure*}   

We will compute the overall \textsl{spectral efficiency ($\se$)} for both phases as  
\begin{equation} 
\se_{\os}=\frac{T^{\tr} }{MT}\sum_{j=2}^{M}\frac{1}{M}\sum_{i=1}^{j-1}R\left(i,j\right)+\frac{T-T^{\tr}}{T}R^{\data}.
\label{equ:SEos}
\end{equation} 
The first and second terms capture the spectral efficiency achieved during and after training, respectively. Rate expression $R\left(i,j\right)$ stands for the downlink rate achieved by User $i$ during cycle $j$. Moreover, $R^{\data}$ is the rate achieved after training, i.e., during half-duplex phase. By only considering the second term in Eq.~\eqref{equ:SEos}, the spectral efficiency of half-duplex counterpart is immediate as
\begin{equation} 
\se_{\hd}=\frac{T-T^{\tr}}{T}R^{\data}.
\label{equ:SEHD}
 \end{equation}
Our objective is to maximize the downlink spectral efficiency. We first quantify the influence of both inter-node and inter-beam interference on $R\left(i,j\right)$ and $R^{\data}$ in Section~\ref{subsec:INIcha} for further analysis. During training, $i-1$ users are served on downlink in Cycle $i>1$. The base station uses power budget $P/M$ to serve each receiving user. 

In this paper, the performance of the following four systems is examined: sequential beamforming strategy with closed and open loop training, half-duplex with closed and open loop training. We differentiate between sequential beamforming and half-duplex systems through the subscripts $\os$ and $\hd$, respectively. The strategy is further detailed by the training type used by the system through another subscript $\mathrm{Cl}$ for closed loop and $\mathrm{Op}$ for open loop training. The superscript is used to denote the system status.  For example, $\se_{\osc}$ stands for the spectral efficiency of a system adopting sequential beamforming strategy with closed loop training.   
\subsection{Rate Performance with Inter-beam and Inter-node Interference}  \label{subsec:INIcha}
To optimize the spectral efficiency of the proposed sequential beamforming strategy, we now quantify the influence of inter-node and inter-beam interference on downlink rate. 
  
  In ZF beamforming, $\textbf{v}_{i}$ is chosen to be orthogonal to other users' channel realization, i.e., $|\textbf{v}_{i}\textbf{h}_{j}|=0 ,j\neq i$. In a genie-aided system where perfect CSI is available, the base station beamforms to users without training and each user receives downlink data at rate $R^{\zf}$ as% at no cost
{  
\begin{equation}
R^{\zf}=\mathbb{E}\left[\log_{2}\left(1+\frac{P}{M} |\textbf{h}_{i}^{*} \textbf{v}_{i}|^2\right)\right]. \label{equ:R_zf}
\end{equation} } 
Rate~\eqref{equ:R_zf} can be viewed as an upper bound for all strategies that employ ZF, since neither training overhead nor inter-beam interference is included. When User $k$ sends the training symbols, the received $\sinr$ of User $i$ ($i<k$) is decided by both inter-beam and inter-node interference, which can be mathematically expressed as
\begin{equation}
\sinr_{i}= 
\frac{ |\textbf{h}_{i}^*\textbf{v}_{i}|^2\frac{P}{M}}{1+\sum_{j\neq i}\frac{P}{M}|\textbf{h}_{i}^*\textbf{v}_{j}|^2+|h_{\mathrm{U_{ik}}}x_{\mathrm{tr}_{k}}|^2} ,\quad i=1,\cdots,k-1 . \label{equ:sinr}
\end{equation}
{  By adopting a Gaussian input, User $i$ with $\sinr_{i}$ can achieve a longer term average rate (over both fading and training error) of
\begin{equation*}
R_{i}= \mathbb{E}\left[\log_{2}\left(1+\sinr_{i}\right)\right] ,\quad i=1,\cdots,k-1.
\end{equation*}}
  We now characterize the downlink rate $R\left(i,j\right)$ with a unified lower bound that is independent of cycle number. In Cycle $k$ of the training phase, the base station beamforms to Users $1,..., k-1$. Each receiving user suffers inter-beam interference from the signals intended to the other $k-2$ users. Our lower bound assumes that users receive additional inter-beam interference from signal intended to other $M-k$ users. Thus, in total, each receiving user suffers inter-beam interference coming from signals to $M-1$ users instead of $k-2$ users. The downlink rate of the receiving users in this scenario is denoted as $R^{\tr}$, which is detailed as
\begin{equation}
R\left(i,j\right)\geqslant R^{\tr}\left(T^{\tr}\right),\quad j=2,...,M,\quad i=1,..., j.
\label{equ:R_tr_LoBd} 
\end{equation}
This lower bound is the same for all the receiving users in all cycles during the training phase. Therefore, the rate expression of~\eqref{equ:SEos} reduces to
\begin{equation} 
\se_{\os}\geqslant\frac{M-1}{2M}\frac{T^{\tr}}{T} R^{\tr}\left(T^{\tr}\right)+\left(1-\frac{T^{\tr}}{T}\right)R^{\data}\left(T^{\tr}\right),
\label{equ:SEosDeatiled} 
\end{equation}
here $T^{\tr}$ is the training duration of sequential beamforming strategy. Comparing to~\eqref{equ:SEHD}, we find in sequential beamforming strategy, on average, the each user utilizes $\frac{M-1}{2M}$ fraction of the training time also to receive downlink data while other users send training signal. 

{  Since the base station is assumed to perform perfect self-interference cancellation, the influence of inter-beam interference is a function of the training method, power, and duration, which is characterized in~\cite{caire2010multiuser}. We now present an extended lemma that quantifies the rate loss due to inter-beam and inter-node interference.} Following the notations in~\cite{caire2010multiuser}, $\Delta R^{\tr}$ and $\Delta R^{\data}$  denote the upper bound of rate gap (compared to perfect zero-forcing) during and after training, respectively.
\begin{Lemma}~\label{lem:INI}
In all cycles of training phase, the downlink data transmission rate of the receiving users, when another user is sending the training symbols are lower bounded as
\begin{align}
R^{\tr}\left(T^{\tr}\right)\geqslant& R^{\zf}- \log\left( \frac{1+ \mathcal{P}_{\mathrm{IBI}}\left(T^{\tr}\right)+\alpha f P}{1 +\frac{\alpha f P}{1+\frac{P}{M}}}\right) \notag \\
=& R^{\zf}-\Delta R^{\tr}\left(T^{\tr}\right), \notag 
\end{align}  
where $ \mathcal{P}_{\mathrm{IBI}}=P(1+fP)^{-\frac{T^{\tr}}{M\left(M-1\right)} }$ and $\frac{P}{M}\frac{M-1}{1+\frac{T^{\tr}}{M} f P}$ for closed and open loop training, respectively.
\end{Lemma} 
\begin{proof}
See Appendix~\ref{App:INI} for detail.
\end{proof}

In the rate gap term $\Delta R^{\tr}$, inter-beam and inter-node interference are reflected through terms $ \mathcal{P}_{\mathrm{IBI}}$ and $\alpha f P$, receptively. If more training symbols are sent, $ \mathcal{P}_{\mathrm{IBI}}$ decreases. This decrease is because that the base station has better CSI estimates, which leads to less inter-beam interference. The inter-node interference term $\alpha fP$ does not change during the whole training phase. It is emphasized that the lower bound present in Lemma~\ref{lem:INI} is independent of the user index and cycle index, due to the use of Eq.~\eqref{equ:R_tr_LoBd}.

{  Many recent works have helped raise the level of self-interference cancellation~\cite{bharadia2014full,duarte2012experiment,Jain2011practicalFD,everett2015measurement,aryafar2012midu,yin2013full,ahmed2013self}. Hence in this work, we assume that full-duplex MIMO array with perfect self-interference cancellation is possible for the base station, which has a larger footprint than a mobile node.  This assumption allows us to focus on characterizing the tradeoff between inter-beam and inter-node interference. Lemma $1$ and other analysis in this paper can be extended to system models that include the impact of limited self-interference cancellation by substituting $fP$, which is the effective uplink training $\snr$ in $ \mathcal{P}_{\mathrm{IBI}}$, as $\frac{fP}{1+ P_{\mathrm{SI}} }$ . Here $P_{\mathrm{SI}}$ is the power of residual self-interference interference. Simulation results for systems with self-interference are provided in Section~\ref{sec:optSE}.}

As $T^{\tr} \to \infty$, the rate loss due to inter-beam interference vanishes and rate gap bound becomes $\log\left( \frac{1+\alpha f P}{1 +\frac{\alpha f P}{1+P/M}}\right)$, which stands for the influence of inter-node interference and is noted as $\Delta R^{\ini}$.   Term $\Delta R^{\ini}$ is as a constant rate loss caused by inter-node interference during the training phase. We will study the impact of this term in the following analysis. Even when $\alpha \to \infty$, the rate loss term is still upper bounded by $\log\left(1+P/M\right)$, which is obviously finite.  This finite rate loss suggests that positive downlink rate gain can still be achieved asymptotically under the influence of inter-node interference, which is later confirmed in Section~\ref{sec:HiSNRAna}.  
 
After training, each user continues to receive data until the end of the block. Thus, only the effect of inter-beam interference exists. We can conveniently obtain the rate expression $R^{\data}$ by setting $\alpha=0$ in Lemma~\ref{lem:INI}, which characterizes inter-beam interference with the help of~\cite{caire2010multiuser}.
\begin{Proposition}\label{pop:IBI}
The downlink transmission rate of User $i$ after training is lower bounded by
\begin{equation}
 R^{\data}\geqslant R^{\zf}-\Delta R^{\data}= R^{\zf}- \log\left( 1+ \mathcal{P}_{\mathrm{IBI}}\right),
\end{equation} 
where $ \mathcal{P}_{\mathrm{IBI}}=P(1+fP)^{-\frac{T^{\tr}}{M\left(M-1\right)} }$ and $\frac{P}{M}\frac{M-1}{1+\frac{T^{\tr}}{M} f P}$ for closed and open loop training, respectively.
 \end{Proposition}
 Similar to the influence of inter-beam interference during training, we find that the influence of inter-beam interference also decreases as training symbols amount increases.
 
 We assume perfect CSI at users by estimating pilots broadcast by each base station antenna. The rate loss due to imperfect CSI at users~\cite{caire2010multiuser} is upper bounded by 
$$\Delta R^{\mathrm{CSIR}} \leqslant \log\left(1+\frac{P/M}{1+\frac{T_{\mathrm{DL}}}{M} P_{\mathrm{DL}}}\right),$$ 
where $T_{\mathrm{DL}}/M$ and $P_{\mathrm{DL}}$ is the number and power of downlink training symbols broadcast by base station antennas.   For example, a system with $M=8$ and $P=15$ dB, with \lq \lq pilots power boost of 3dB\rq \rq \space from~\cite{3gpp.36.213}, when $1$ training pilot is used for each base station antenna,  the rate loss is upper bounded by 0.06 bps/Hz, which is smaller than $3\%$ of the associated $R^{\zf}$. This rate loss is the same for half-duplex and the proposed system. {  Similar to~\cite{caire2010multiuser,kobayashi2011training}, the rate considered in this paper is the {\sl ergodic} rate, which can be achieved by spanning codewords across enough large number of blocks.}
\section{Training Time Optimization} \label{sec:optLen}
In sequential beamforming strategy described in Section~\ref{sec:ABprop}, the base station obtains CSI from uplink training symbols. %While more training symbols helps designing a more accurate precoder and reduce inter-beam interference, longer training also implies higher overall inter-node interference. 
Obtaining the best spectral efficiency performance requires a balance between inter-beam and inter-node interference by optimizing training duration. In this section, we analyze the optimal training duration for sequential beamforming and traditional half-duplex strategy with both closed and open loop training. {  We use the metric of spectral efficiency to characterize the  optimal solution for the proposed strategy.} The optimization 
\begin{equation}\label{eq:OptSta}
T^{\tr*}_{s}=\arg\max \se_{s}\left(T^{\tr}\right),  
\end{equation}
is implemented for $s= \osc$, $\oso$, $\hdc$ and $\hdo $. We use the superscript $*$ to denote optimal solution. The optimality in this paper is under the criterion of maximum spectral efficiency. We also consider $T^{\tr}$ to be continuous. The accent $\widetilde{\quad}$ is used to represent approximation in Sections~\ref{sec:optLenAB},~\ref{sec:optLenHD}, where closed form analytical solutions are not feasible. 
 \subsection{Optimal Training Duration of Sequential Beamforming}\label{sec:optLenAB}
In this subsection, we solve the optimization problem posed in~\eqref{eq:OptSta} by applying a \textsl{Marginal Analysis}~\cite{baumol1977economic} technique. As shown below, the marginal analysis allows accurate closed form approximation for systems with both closed and open loop training.
\begin{Proposition} \label{obs:MCMUAB}
The optimal training duration of sequential beamforming strategy happens at the point where the spectral benefit of adding training symbols equals to loss, i.e.,
\begin{equation}
\frac{\partial \se_{\os}\left(T^{\tr *}\right)}{\partial T^{\tr}}=0.\label{equ:optLenABdef}
\end{equation}
\end{Proposition}
\begin{proof}
Since the mobile nodes are half-duplex, more training implies less time for downlink data reception. The influence of inter-beam interference on the rate in Proposition~\ref{pop:IBI} suggests that the increase in rate with respect to training increase is monotonically decreasing. Combining with the facts above, we conclude that the benefit in spectral efficiency from $M$ additional training symbols is monotonically decreasing as training grows. From Lemma~\ref{lem:INI}, the influence of longer inter-node interference duration is monotonically increasing. Thus, the spectral efficiency $\se_{\os}$ is concave in $T^{\tr}$. Therefore, a unique $T^{\tr *}_{\os}$ exists to optimize the spectral efficiency.  
\end{proof}

Applying Taylor's expansion to~\eqref{equ:optLenABdef} and ignoring all the expansion terms yield  
%Since solving~\eqref{equ:optLenABdef} is challenging, we then apply Taylor's expansion to and ignore all the expansion terms, which yields  
\begin{equation}\label{equ:optLenABcom}
\se_{\os}\left(\widetilde{T}^{\tr *}_{\os}\right)\approx \se_{\os}\left(\widetilde{T}_{\os}^{*}+M\right).
\end{equation}
With the help of spectral efficiency characterization provided in Lemma~\ref{lem:INI}, expanding both sides of~\eqref{equ:optLenABcom} leads to
\begin{align}
&\frac{M-1}{2M}\frac{\widetilde{T}^{\tr}_{\os}}{T}\left[R^{\tr}\left(\widetilde{T}^{\tr}_{\os}+M\right)-R^{\tr}\left(\widetilde{T}^{\tr}_{\os}\right)\right]\notag\\ 
&+\frac{T-\widetilde{T}^{\tr}_{\os}}{T}\left[R^{\data}\left(\widetilde{T}^{\tr}_{\os}+M\right)-R^{\data}\left( \widetilde{T}^{\tr}_{\os} \right) \right]\notag \\
=&\frac{M-1}{2T}\left[R^{\data}\left(\widetilde{T}^{\tr}_{\os}+M\right)-R^{\tr}\left(\widetilde{T}^{\tr}_{\os}+M\right)\right]\notag\\ 
&+\frac{M+1}{2T} R^{\data}\left(\widetilde{T}^{\tr}_{\os}+M\right). \label{equ:optABcom1} 
\end{align}
The left side in~\eqref{equ:optABcom1} is the benefit obtained in spectral efficiency by adding $M$ training symbols. We note this benefit as \textsl{Marginal Utility} (MU).  The MU comes from the fact: more training can reduce inter-beam interference both during and after training, which corresponds to the first and second term on left side of~\eqref{equ:optABcom1}, respectively.

More training symbols results in lower inter-beam interference in half-duplex phase. By using Proposition~\ref{pop:IBI}, it is expressed as a rate increase of
\begin{align}
R^{\data}&\left( \widetilde{T}^{\tr}_{\os}+M\right)-R^{\data}\left( \widetilde{T}^{\tr}_{\os} \right)\notag\\&=\log\left(1+\frac{   \mathcal{P}_{\mathrm{IBI}}\left(\widetilde{T}^{\tr}_{\os}\right)  - \mathcal{P}_{\mathrm{IBI}}\left( \widetilde{T}^{\tr}_{\os}+M\right)  }{1+ \mathcal{P}_{\mathrm{IBI}}\left( \widetilde{T}^{\tr}_{\os}+M\right)  } \right).\label{equ:deltaRibi}
\end{align}
We refer the rate improvement due to less inter-beam interference as $ \delta R^{\data}\left(\widetilde{T}^{\tr}_{\os}\right)$. In the same spirit, the rate increase of $R^{\tr}$ by lower inter-beam interference during training phase is 
\begin{equation}
R^{\tr}\left(\widetilde{T}^{\tr}_{\os}+M\right)-R^{\tr}\left(\widetilde{T}^{\tr}_{\os}\right)
\approx  \delta R^{\data}\left(\widetilde{T}^{\tr}_{\os}\right). \notag
\end{equation}
We find that, the rate improvement due to less inter-beam interference is almost constant during and after training. Applying the two results above, the marginal utility is then
\begin{align}
MU&= \left(1-\frac{M+1}{2M}\frac{\widetilde{T}^{\tr}_{\os}}{T}\right)\delta R^{\data}\left(\widetilde{T}^{\tr}_{\os}\right)\notag \\&\approx  \left(1-\frac{1}{2}\frac{\widetilde{T}^{\tr}_{\os}}{T}\right)\delta R^{\data}\left(\widetilde{T}^{\tr}_{\os}\right), \label{equ:MU}
\end{align}
which suggests a rate increase of $\delta R^{\data}\left(\widetilde{T}^{\tr}_{\os}\right)$ in $ 1-\frac{1}{2}\frac{\widetilde{T}^{\tr}_{\os}}{T} $ fraction of the whole block is achieved by adding $M$ training symbols. Later we further obtain the marginal utility of half-duplex counterparts by the same process with $\se_{\os}$ substituted as $\se_{\hd}$.

On the right side of~\eqref{equ:optABcom1} is the loss of spectral efficiency, referred to as \textsl{Marginal Cost} (MC), due to longer training and comprises two parts. 
The first term corresponds to the fact that additional inter-node interference is suffered in $\frac{M-1}{2}$ of the $M$ symbols. The second term reflects that the rest $\frac{M+1}{2}$ training symbols are still not able to be utilized for downlink. With the help of Lemma~\ref{lem:INI} and Proposition~\ref{pop:IBI}, the rate loss due to additional inter-node interference is 
\begin{align*}
&R^{\data}\left(\widetilde{T}^{\tr}_{\os}+M\right)-R^{\tr}\left(\widetilde{T}^{\tr}_{\os}+M\right)\\
&= \log\left( \frac{1+\frac{\alpha f P}{1+ \mathcal{P}_{\mathrm{IBI}}\left(\widetilde{T}^{\tr}_{\os}+M\right)}}{1 +\frac{\alpha f P}{1+\frac{P}{M}}}\right) \approx  \Delta R^{\ini}.
\end{align*}
The downlink rate loss due to inter-node interference in the training phase is almost independent of the training duration. The downlink rate is immediate as $R^{\zf}$. The marginal cost is then
\begin{equation}
MC=\frac{M-1}{2T}  \Delta R^{\ini}+\frac{M+1}{2T} R^{\zf} \approx \frac{M}{2T} \left[\Delta R^{\ini}+ R^{\zf}\right] \label{equ:MC} ,
\end{equation}
which is independent of training symbol amount. The approximation made in Eq.~\eqref{equ:MC} holds for large $T$. In the training phase, Eq.~\eqref{equ:MC} suggests that, on average, each user receives downlink data during half of the training time.

The unique optimal point $T^{\tr *}_{\os}$ happens at the point where the spectral efficiency benefit (marginal utility) and cost (marginal cost) break even, i.e., $MU=MC$ . Using~\eqref{equ:MU} and~\eqref{equ:MC}, it is mathematically captured as
\begin{align}
\left(1-\frac{M+1}{2M}\frac{\widetilde{T}^{\tr}_{\os}}{T}\right)&\delta R^{\data}\left(\widetilde{T}^{\tr}_{\os}\right)\notag \\
&=\frac{M-1}{2T}  \Delta R^{\ini}+\frac{M+1}{2T} R^{\zf}. \label{equ:optLenABfin} 
\end{align} 
The optimal training duration of sequential beamforming strategy with closed and open loop training is then obtained by the inter-beam interference characterization provided in Proposition~\ref{pop:IBI}. 
\begin{theorem}\label{thm:LenABC}
The approximation of optimal training duration $T_{\osc}^{\tr *}$ of sequential beamforming strategy with closed loop training is 
\begin{equation}
\widetilde{T}_{\osc}^{\tr *}=M\left(M-1\right)\frac{\log\left(\frac{TP}{c}\right)+\log\left(\left(1+fP\right)^{\frac{1}{M-1}}-1\right)}{\log\left(1+fP\right)},
\end{equation}
where $c=\frac{M-1}{2}\Delta R^{\ini}+\frac{M+1}{2} R^{\zf}$.
\end{theorem}
\begin{proof}  
See Appendix~\ref{App:OptLenABC}.
\end{proof} 
 
Several interesting observations are made here. First, as $T$ grows, for closed loop training based sequential beamforming strategy, the optimal training duration scales as $ \log T $. We later observe similar scaling law for its half-duplex counterpart. This scaling law, to our best knowledge, has not been reported before. Second, as inter-node interference becomes stronger, less training is sent to account for the higher training \emph{cost}. Third, the optimal training symbols amount scales as $ \frac{\log P}{\log\left(1+fP\right)} $ with respect to $P$, which is less than $ \log P $. From Theorem 4 in~\cite{LimitedFb2006Jindal}, we conclude that full multiplexing is not obtained as $P$ grows. Fourth, the number of training symbols increases almost quadratically with respect to the number of users $M$, which lies well with the intuition that training symbols number scales with the number of total channel coefficients.
 
Fig.~\ref{fig:optLenC} provides both optimal training duration and its approximation of sequential beamforming strategy with closed loop training. Since optimal training duration scales as $ \log T $, the fraction of training duration scales as $ \frac{\log T}{T}$ and is further confirmed numerically.
%\vspace*{-0.5\baselineskip}
  \begin{figure*}[htbp]
  \centering
    \subfloat[Closed Loop Systems\label{fig:optLenC}]{    
	\includegraphics[scale=1]{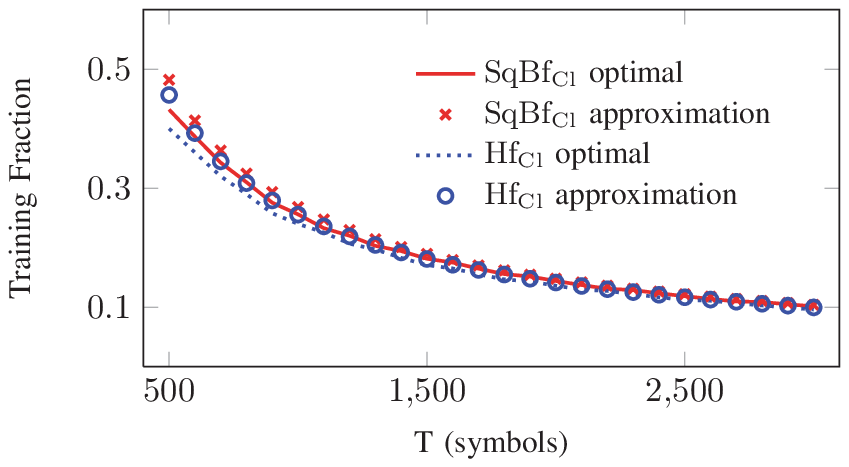} 
    }
    \hfill
    \subfloat[Open Loop Systems\label{fig:optLenO}]{
	\includegraphics[scale=1]{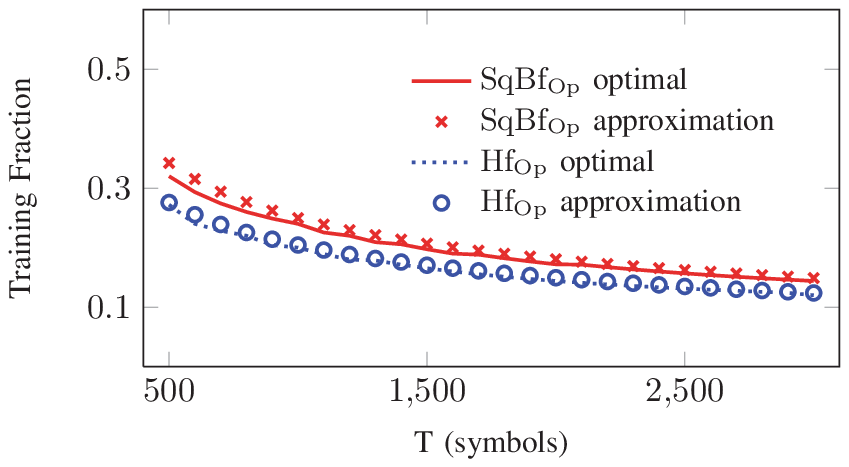} 
    }
 \caption{Optimal training duration fraction of $8\times 8$ sequential beamforming and half-duplex strategy with closed and open loop training at $P=15$ dB with $f=0.1$ and $\alpha=0.3$.}
    \label{fig:optLen} 
  \end{figure*}  
  %\vspace*{-0.5\baselineskip} 
In the same spirit, the optimal training duration of sequential beamforming strategy with open loop training is obtained as below.  
\begin{theorem}\label{thm:LenABO}
The approximation of optimal training duration $T_{\oso}^{\tr *}$ of sequential beamforming strategy with open loop training is 
\begin{equation}
\widetilde{T}_{\oso}^{\tr *}=\sqrt{\frac{\left(M-1\right)T}{f \frac{(M-1)\Delta R^{\ini} +(M+1)R^{\zf}}{2M}}} \approx \sqrt{\frac{\left(M-1\right)T}{f \frac{\Delta R^{\ini} +R^{\zf}}{2}}}. \notag
\end{equation} 
\end{theorem}
 \begin{proof}  
See Appendix~\ref{App:OptLenABO}.
\end{proof}
 
In Theorem~\ref{thm:LenABO}, we observe that for large $T$, the optimal training duration scales as $ \sqrt{T} $. The optimal fraction of time resource devoted to training then decreases as $ \frac{1}{\sqrt{T}}$, which is slower than that of closed loop training systems.  The scaling rate of $ \sqrt{T} $ has been observed in various open loop training based systems. For example, similar scaling has been observed for both half-duplex MIMO broadcast channels with analog training\cite{kobayashi2011training} and point-to-point MIMO\cite{hassibi2003much}. Such scaling rate has also been observed in MIMO downlink with the full-duplex base station and full-duplex node~\cite{du2015SPAWC}. We find the same scaling law is shared by open loop training based systems. In Section~\ref{sec:optLenHD}, we conclude that the half-duplex counterparts also follow the respective scaling laws. Numerical results presented in Fig.~\ref{fig:optLenO} confirm our observation.

For sequential beamforming strategy with open loop training, as the number of transmitting antennas $M$ increases, the optimal training duration scales as $ \sqrt{M} $, unlike $ M^2 $ scaling in closed loop training based systems. The slower scaling rate in open loop training systems suggests a lower overhead cost in systems with a large number of users.  Further analytical results in Section~\ref{sec:HiSNRAna} confirm this observation.
 
Similar to closed loop systems, as inter-node interference increases, larger rate loss during training is expected in open loop training systems. Thus, one should use fewer symbols for training to account for this effect. Another interesting finding is that even when no inter-node interference exists, the optimal training duration is not $T$. The reason is that sequential beamforming strategy is only able to recover the training overhead partially.   
\subsection{Optimal Training Duration of Half-duplex Strategy} \label{sec:optLenHD}
In this subsection, we apply the marginal analysis method developed in Section~\ref{sec:optLenAB} to obtain approximations of optimal training duration of half-duplex systems. As a by-product of analysis in section \ref{sec:optLenAB}, we find the marginal utility, which stands for the gain in spectral efficiency of adding $M$ more training symbols, for half-duplex systems is
\begin{align}
MU=&\frac{T-\widetilde{T}^{\tr *}_{\hd}}{T}\left[R^{\data}\left( \widetilde{T}^{\tr *}_{\hd}+M\right)-R^{\data}\left( \widetilde{T}^{\tr *}_{\hd} \right)\right]\notag \\
=& \frac{T-\widetilde{T}^{\tr *}_{\hd}}{T} \delta R^{\data}\left(\widetilde{T}^{\tr *}_{\os}\right).  \label{equ:optHDMU} 
\end{align} 
 The marginal cost of half-duplex strategy is conveniently obtained by ignoring inter-beam interference after training as
\begin{equation}
MC=\frac{M}{T}R^{\zf}. \label{equ:optHDMC} 
\end{equation}
 
The approximation is obtained by letting marginal cost and utility be equal in half-duplex systems. We further proceed by applying the rate characterization provided in Proposition~\ref{pop:IBI}. The result regarding closed loop training based systems is first presented with open loop result follows.
\begin{theorem}\label{thm:LenHDC}
The approximation of optimal training duration $T_{\hdc}^{\tr *}$ of half-duplex strategy with closed loop training is
\begin{align*}
\widetilde{T}_{\osc}^{\tr *}=M\left(M-1\right)\frac{\log\left(\frac{TP}{c}\right)+\log\left(\left(1+fP\right)^{\frac{1}{M-1}}-1\right)}{\log\left(1+fP\right)},
\end{align*} 
where $c=MR^{\zf}$.
\end{theorem}
 \begin{proof}  
See Appendix~\ref{App:OptLenHDC}.
\end{proof} 

We observe the optimal training duration $\widetilde{T}_{\osc}^{\tr *}$ and fraction $\frac{\widetilde{T}_{\osc}^{\tr *}}{T}$ of half-duplex counterpart to share the same scaling law as sequential beamforming strategy in Theorem~\ref{thm:LenABC}. It should also be noted that as the number of antenna $M$ increases, the optimal number of training symbols also increases quadratically. Comparing to Theorem~\ref{thm:LenABC}, the only difference lies in the $\log(c)$ term in the numerator, which can be viewed as the normalized marginal cost of the strategy. This finding also suggests that the optimal training duration difference between sequential beamforming strategy and half-duplex is a constant gap which is independent of the block length. Therefore, the difference in the fraction of training time decreases as $T$ increases. 
\begin{theorem} \label{thm:LenHDO}
The approximation of training duration $T_{\hdo}^{\tr *}$ that optimizes spectral efficiency of open loop training half-duplex systems is
 \begin{equation}
\widetilde{T}_{\hdo}^{\tr *}= \sqrt{\frac{\left(M-1\right)T}{f R^{\zf}}}.
\end{equation} 
\end{theorem}
\begin{proof}
See Appendix~\ref{App:OptLenHDO}.
\end{proof}

Similar to sequential beamforming system with open loop training, the optimal training duration scales with $T$ and $M$ at the rate of $\sqrt{T}$ and $\sqrt{M-1}$, respectively, as block length and antennas number grows. It should also be noted that by substituting the normalized marginal cost term $\frac{\Delta R^{\ini} +R^{\zf}}{2}$ as the half-duplex system's normalized marginal cost term $R^{\zf}$, we can also obtain Theorem~\ref{thm:LenHDO}. Instead of assuming each user has the same power constraint $P$ of the base station~\cite{kobayashi2011training}, our approximation results further consider the limitation of user power.  
%\begin{remark}
%By comparing the training duration of closed and open loop systems, we find that \textsl{training type} dominates the scaling of optimal training time with respect to both block length and number of antennas.  For example, as block length $T$ grows, optimal training duration for both sequential beamforming and half-duplex systems with open loop training scale as $\sqrt{T}$ while the closed loop training counterparts scale as $\log T$.
%\end{remark}
The closed-form approximations are further applied in Section~\ref{sec:optSE} to characterize the spectral efficiency of sequential beamforming and half-duplex strategy with optimal training duration.
\section{Spectral Efficiency Evaluation}\label{sec:optSE}
In training based multiuser MIMO downlink systems, spectral efficiency is reduced due to imperfect CSI and training overhead resulting from its acquisition. To quantify the spectral efficiency loss of different systems, we compare the spectral efficiency of different systems with optimal training duration to a system where perfect CSI is available for the base station at no cost. It can be visualized as a genie provides perfect CSI to base station at the beginning of each block. Thus it serves as an upper bound for systems' performance with ZF; we label the perfect CSI system as \textsl{genie-aided system}. The spectral efficiency achieved is $\se^{\zf}=R^{\zf}$, which is presented in~\eqref{equ:R_zf}. The rate loss due to training overhead is then
\begin{equation}
\Delta \se_{s}= \se^{\zf}- \se_{s}, \quad s\in\{\osc,\oso,\hdc,\hdo\}.
\end{equation}
The spectral efficiency of the half-duplex counterparts are also analyzed for comparison.
\begin{theorem}\label{thm:SEABC}
The spectral efficiency loss of closed loop training based sequential beamforming system with respect to genie-aided system is upper-bounded as
%\leqslant& \se^{\zf}-\se_{\osc}\left(\widetilde{T}_{\osc}^{\tr *}\right) \notag\\
\begin{align}
\Delta  \se_{\osc}&\left(T_{\osc}^{\tr *}\right) \notag \\
\leqslant &  \left(\frac{M-1}{2M}\Delta R^{\ini}+\frac{M+1}{2M} R^{\zf}\right)\frac{M\left(M-1\right)}{\log\left(1+fP\right)}\frac{\log T}{T}\notag\\&+o\left(\frac{\log T}{T}\right).  \label{equ:dSEABC}
\end{align} 
\end{theorem}
\begin{proof}
See Appendix~\ref{App:OptSEABC}.
\end{proof} 

Here $o\left(\frac{\log T}{T}\right)$ is a term that vanishes as $T$ increases, i.e., $\lim_{T \to \infty}\frac{o\left( \log T/T\right)}{ \log T/T}=0 $. By employing higher training power $f$, or by using longer block length $T$, the spectral efficiency overhead  decreases. However, in a more realistic scenario where user power and block length are inherently limited, the spectral efficiency loss cannot be neglected. Based on expression~\eqref{equ:dSEABC}, some observations are made for sequential beamforming strategy with closed loop training. i) The spectral efficiency loss scales quadratically as $M$ increases, which indicates sequential beamforming strategy with closed loop training is not a good choice for systems with a large number of antennas. ii) The spectral efficiency loss decreases rapidly as $ \frac{\log T}{T} $ as $T$ increases.  Fig.~\ref{fig:SE_T} presents the spectral efficiency policy for different strategy versus $T$. We observe that as $T$ grows, spectral efficiency loss drops rapidly for systems with closed loop training, which agrees with our analysis.  iii) As inter-node interference level decreases, smaller term $\Delta R^{\ini}$ suggests less spectral efficiency loss which is confirmed in Fig.~\ref{fig:SE_T}. 

\begin{figure*}[htbp]
  \centering
    \subfloat[Closed Loop Systems\label{fig:SE_T_C}]{%
	\includegraphics[scale=1]{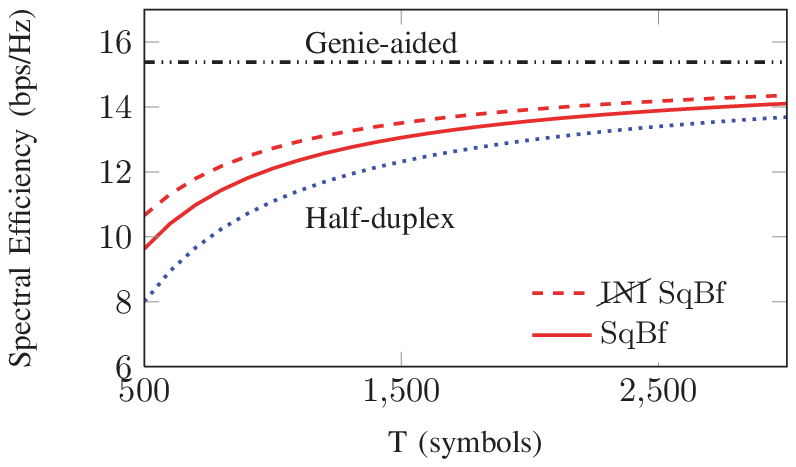} 
    }
    \hfill
    \subfloat[Open Loop Systems\label{fig:SE_T_O}]{%
	\includegraphics[scale=1]{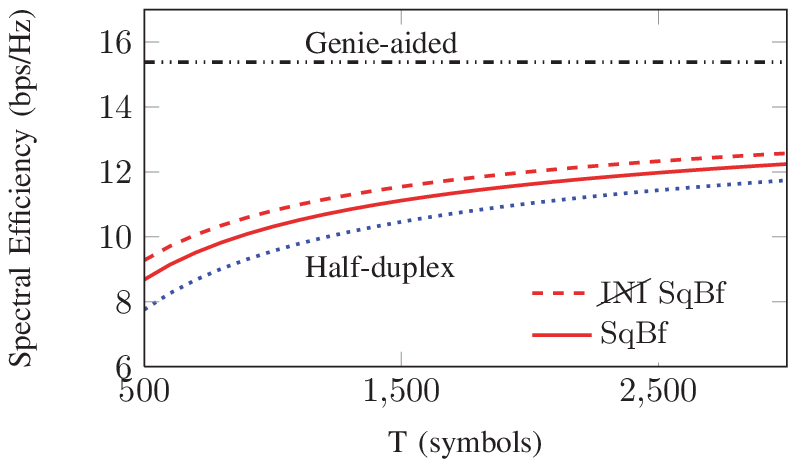} 
    }
 \caption{Spectral efficiency of an $8\times 8$ system with optimal training duration with and without inter-node interference and half-duplex counterparts with $P=15$ dB for $f=0.1$.}
    \label{fig:SE_T} 
  \end{figure*}    
\begin{theorem}\label{thm:SEABO}
The spectral efficiency loss of open loop training based sequential beamforming system with respect to genie-aided system is upper-bounded as
%\leqslant& \se^{\zf}-\se_{\oso}\left(\widetilde{T}^{\tr *}_{\oso}\right) \notag\\
\begin{align} 
&\Delta  \se_{\oso}\left(T^{\tr *}_{\oso}\right) \notag \\
&\leqslant 2\sqrt{\frac{(M-1)\left[\frac{M-1}{2M}\Delta R^{\ini} +\frac{M+1}{2M}R^{\zf}\right]}{fT}}+o\left(\frac{1}{\sqrt{T}}\right). \label{thm:dSEABO}
\end{align}
\end{theorem}
\begin{proof}
See Appendix~\ref{App:OptSEABO}.
\end{proof} 

In~\eqref{thm:dSEABO}, the term $o(\frac{1}{\sqrt{T}})$ shows that the additional spectral efficiency loss term vanishes in systems with a large $T$. Interestingly, we observe a different scaling law with respect to both block length and antenna number. For sequential beamforming with open loop training, the spectral efficiency loss grows only at the rate of $ \sqrt{M-1} $, which is slower than $ M(M-1) $ in Theorem~\ref{thm:SEABC}. Thus, for systems with a large number of users, sequential beamforming with open loop training is advisable. 

On the other hand, the spectral efficiency loss decreases as $  \frac{1}{\sqrt{T}}$, which is confirmed from Fig.~\ref{fig:SE_T}. It should be further noted that the decreasing rate (w.r.t. $T$) is slower than that of systems with closed loop training ($\log T/T$). From Fig.~\ref{fig:SE_T}, an increase spectral efficiency is achieved by sequential beamforming strategy at low inter-node interference level. 
%\textsl{Power Controlled Sequential Beamforming}: In the discussion above, we assumed that the base station transmits to users with power $P/M$ in the training phase. Based on the observation that base station would not beamform to the users whose CSI have not been yet collected, one can further increase the spectral efficiency by letting base station adapt power in currently served users, i.e., for whom the CSI has been collected. We refer this strategy as \textsl{power controlled sequential beamforming} (noted as $\mathrm{PwCo}$ $\mathrm{SqBf}$ in legend).

Fig.~\ref{fig:SE_T} plots closed and open loop training based systems with power controlled sequential beamforming, sequential beamforming, and half-duplex strategy. We observe a further spectral efficiency increase by allowing power adaptation during training. 
%However, the gain does requires increased complexity as remarked below.
%
%\begin{remark} 
%In sequential beamforming strategy, the lower bound for the downlink rate during training is the same for all receiving users in all cycles. Hence, both base station and users use the same modulation rate during full-duplex transmission phase. However, in power controlled sequential beamforming strategy, since base station adapts transmission energy to serve different amount of users in each cycle. Receiving users, then, need to adjust receiving rate at the end of each cycle, which leads to a much higher complexity in systems with large numbers of users. 
%\end{remark}
Having established performance bounds for the spectral efficiency loss of the sequential beamforming policy, we now investigate the performance of the half-duplex counterpart to compute the gains of the proposed sequential beamforming strategy.
\begin{theorem}\label{thm:SEHDC}
The spectral efficiency loss of closed loop training based half-duplex systems with respect to  genie-aided system is upper bounded as  
%\leqslant \se^{\zf}-\se_{\hdc}\left(\widetilde{T}_{\hdc}^{\tr *}\right)
\begin{equation}
\Delta  \se_{\hdc}\left(T_{\hdc}^*\right)\leqslant R^{\zf} \frac{M\left(M-1\right)}{\log\left(1+fP\right)}\frac{\log T}{T}+o\left(\frac{\log T}{T}\right).
\end{equation} 
\end{theorem}
\begin{proof}
See Appendix~\ref{App:OptSEHDC}.
\end{proof}  

Here we observe the same scaling of spectral efficiency loss with respect to the number of antennas $M$ and block length $T$ as in sequential beamforming strategy with closed loop training. Actually, we can obtain Theorem~\ref{thm:SEHDC} by replacing the normalized marginal cost term $ \frac{M-1}{2M}\Delta R^{\ini}+\frac{M+1}{2M} R^{\zf} $ in Theorem~\ref{thm:SEABC} with $R^{\zf}$. The main reason is the similarity between the marginal utility term in sequential beamforming and half-duplex system.
\begin{theorem}\label{thm:SEHDO}
The spectral efficiency loss of open loop training based half-duplex systems with respect to  genie-aided system is upper bounded as  
%\leqslant \se^{\zf}-\se_{\hdo}\left(\widetilde{T}_{\hdo}^{\tr *}\right)
\begin{equation}
\Delta  \se_{\hdo}\left(T_{\hdo}^*\right)\leqslant 2\sqrt{\frac{(M-1) R_{\zf}}{fT}}.
\end{equation} 
\end{theorem}
\begin{proof}
See Appendix~\ref{App:OptSEHDO}.
\end{proof}  

The spectral efficiency loss scaling with both block length $T$ and the number of antennas $M$ are identical to that of sequential beamforming system with open loop training. Theorem~\ref{thm:SEHDO} can  be viewed as changing the normalized marginal cost of sequential beamforming strategy into its half-duplex counterpart.
 \begin{figure*}[htbp]
  \centering
    \subfloat[Closed Loop Systems]{%
	\includegraphics[scale=1]{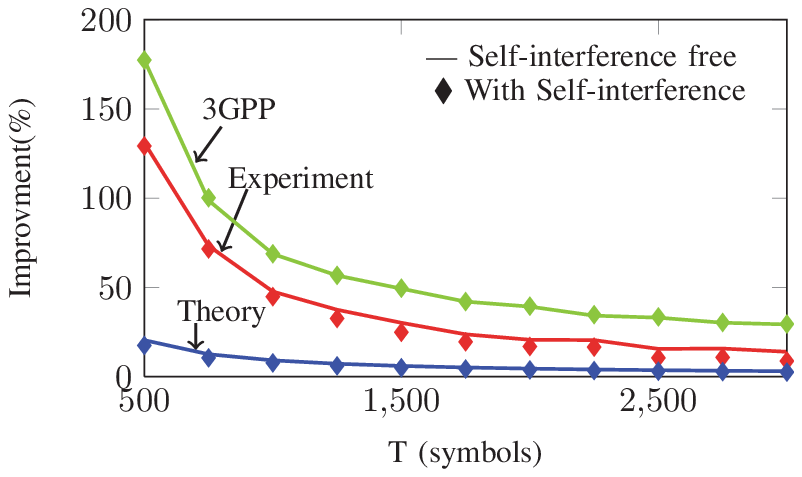} 
    }
    \hfill
    \subfloat[Open Loop Systems]{%
	\includegraphics[scale=1]{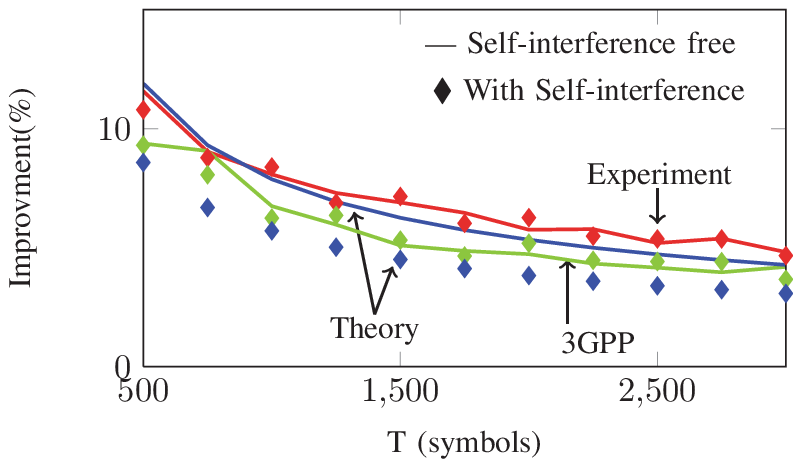} 
    }
 \caption{ {  Percent spectral efficiency improvement of an $8\times 8$ system with sequential beamforming strategy for systems with and without self-interference at $P=15$ dB with $f=0.1$ and $\alpha=0.3$. Here experiment and 3GPP refer to the simulation results obtained with experimental data~\cite{everett2015measurement} and 3GPP 3D channel model~\cite{3gpp.36.873}, respectively. Theory refers to results in Theorem 5-8.}}
    \label{fig:SE_Imp} 
  \end{figure*} 
Comparing Theorem~\ref{thm:SEHDC} and Theorem~\ref{thm:SEHDO} to their sequential beamforming counterparts, spectral efficiency loss is  substantially reduced by adopting sequential beamforming strategy. Sequential beamforming strategy improves spectral efficiency performance significantly. 

{   We now further validate the sequential beamforming strategy for an $8 \times 8$ system with experimental data from~\cite{everett2015measurement} and 3GPP 3D channel model~\cite{3gpp.36.873}. The number of base station antennas is chosen to be $8$, which is the maximum number currently supported by LTE. In the experiment, the authors of~\cite{everett2015measurement} measure the channel realization between an $8 \times 9$ two-dimensional antenna array and $12$ randomly located users.  The measurement is conducted in both indoor and outdoor environment. In 3GPP channel generation, the users are randomly uniformly placed in Urban Macrocell with equal probability to be indoor and outdoor. The simulations for systems containing self-interference is modeled as follows. The base station has a transmit power budget of $0$ dBm and a noise floor of $-90$ dB. Self-interference is managed through a combination of transmitter-receiver isolation of $40$ dB,  analog domain cancellation of $30$ dB, and digital domain cancellation of $30$ dB.}

{  The spectral efficiency of both sequential beamforming and half-duplex counterparts are evaluated through Monte Carlo simulations with $10,000$ iterations for both the proposed strategy and half-duplex counterpart. In each iteration, $8$ random users and $8$ antennas in a random horizontal antenna array are simulated.  For closed loop training, feedback bits are equally divided into real and imaginary parts with $3$ bits for integer part and rest for the fractional part.  In the simulation, we assume that the downlink transmit power in the training phase is adapted, which is not allowed in the previous theoretical analysis due to the analytical intractability in optimal power allocation when imperfect CSI exists. Thus, in cycle $i>1$, when Users $1,..,i-1$ receive data on downlink, each receiving user signal will be precoded with power constraint $\frac{P}{i-1}$. }

Fig.~\ref{fig:SE_Imp} confirms the spectral efficiency improvement achieved by sequential beamforming. The spectral efficiency improvement achieved in Fig.~\ref{fig:SE_T} is also shown for reference. Similar to results in Fig.~\ref{fig:SE_T}, sequential beamforming demonstrates a significant spectral efficiency improvement.

  For example, in a typical LTE system, there are around $500$ to $2100$ symbols in each slot depending on the available bandwidth ($1.4$ MHz to $5$ MHz). When the block length equals $500$ symbols, proposed sequential beamforming strategy attains an over 130\% and 12\% spectral improvement under the influence of inter-node interference for closed and open loop training systems, respectively. As $T$ grows, the performance of half-duplex counterparts grows. Thus the improvement in sequential efficiency decreases. From Fig.~\ref{fig:SE_T}, we conclude that a notable spectral efficiency improvement is still observed even for systems with long block length ($T=3000$). Lower inter-node level does show a better spectral efficiency improvement in Fig.~\ref{fig:SE_Imp}. 
\begin{remark}
For closed loop systems, sequential beamforming demonstrates a higher spectral efficiency compared to the results in Fig.~\ref{fig:SE_T_C}, where the base station serves each downlink users at a fixed power $P/M$ in the training phase.  This improvement suggests that proper power adaptation can increase the performance of sequential beamforming dramatically. On the other hand, we find power adaptation does not influence the spectral efficiency improvement of open loop system.
\end{remark}   

In this section, significant spectral efficiency improvement by adopting sequential beamforming is observed. %Power adaptation influences the spectral efficiency of closed loop systems significantly. 
{  In Section~\ref{sec:HiSNRAna}, we compare the spectral efficiency asymptotically where equal power allocation and ZF can achieve the full multiplexing gain.} %As a byproduct, a comparison between closed and open loop training method in high $\snr$ regime is also presented. 
\section{High $\snr$ Analysis}  \label{sec:HiSNRAna}
In Section~\ref{sec:optLen} and Section\ref{sec:optSE}, with optimized training duration, sequential beamforming strategy exhibits significant spectral efficiency improvement in the finite $\snr$ regime.  We continue our investigation of sequential beamforming strategy in the high $\snr$ regime where equal power allocation and ZF can achieve the full multiplexing gain.  Notation $\doteq$ is used to denote exponential equality, i.e.,
$$g\left(P\right) \doteq P^\zeta \Leftrightarrow  \mathrm{\lim_{P \to \infty}}\frac{\log g\left(P\right)}{\log P}=\zeta.$$ 
%As a byproduct of our analysis, we  observed that closed loop training based systems and open loop training differ both in optimal training duration and spectral efficiency.
Since $fP\doteq P$, we now assume the power constraint for training is $P^\zeta$ to account for the limitation of training power.  We use a multiplexing gain metric $r$, which can be mathematically captured as
\begin{equation}
 \lim_{P \to \infty} {\frac{\se_{S}\left(\zeta,T^{\tr}\right)}{\log P}} \doteq r_{s} , \quad s\in\{\osc,\oso,\hdc, \hdo \}. 
\end{equation}
Our objective is to maximize the spectral efficiency asymptotically under certain training power constraint, which is mathematically captured as, for $s= \osc$, $\oso$, $\hdc$ and $\hdo $,
\begin{equation}
\max_{T^{\tr}} \quad r_{s}\left(\zeta,T^{\tr}\right).
\end{equation}

We first present the results regarding sequential beamforming system with closed loop training. The results for sequential beamforming strategy with open loop training then follows. In the asymptotic characterization of sequential beamforming strategy, for mathematical concision, we consider the fraction of full-duplex transmission term $\frac{M-1}{2M}$ in~\eqref{equ:SEosDeatiled} to be $\frac{1}{2}$. This approximation is valid for systems with large numbers of antennas.   

\subsection{Sequential Beamforming with Closed Loop Training}\label{sec:HighSNRC}
In this subsection, we consider the relationship between multiplexing gain $r$ and training power constraint $\zeta$. Similar to the approach in the finite $\snr$ regime; we first present a lemma capturing the influence of inter-beam and inter-node interference in the high $\snr$ regime, then the spectral efficiency is characterized. We define $\theta=\left(M\left(M-1\right)\right)/T$, which is useful in analysis.

\begin{Lemma}\label{lem:HiINI}
In closed loop systems, the downlink data transmission rate during training phase, under the influence of inter-beam and inter-node interference, is
\begin{equation}
\lim_{P \to \infty}\frac{ R^{\tr}\left(T^{\tr}\right)}{\log P}=\max\left(\min\left( \frac{\zeta}{\theta} \frac{T^{\tr}}{T},1-\zeta\right),0\right).
\end{equation} 
\end{Lemma}
\begin{proof} 
The proof is obtained by substituting the training power in Appendix~\ref{App:INI} $fP$ with $P^{\zeta}$.
\end{proof}

Interestingly, we observe the impact of inter-node interference in the high $\snr$ regime to be divided into two scenarios. If only coarse CSI is available, the influence of inter-beam interference dominates the rate performance during training, i.e., there is no impact of inter-node interference on performance. Otherwise, the influence of inter-node interference dominates the rate performance during training. Now we present the maximal multiplexing gain as a function of training power constraint $\zeta$ for different closed loop training systems. 

Applying Lemma~\ref{lem:HiINI} to characterize~\eqref{equ:SEosDeatiled}, we have
\begin{align} 
\lim_{P \to \infty}\frac{\se_{\osc}}{\log P}  
=& \frac{1}{2}\frac{T^{\tr}}{T} \max\left( \min\left(\frac{ T^{\tr}}{T}\frac{\zeta}{\theta},1-\zeta\right),0\right) \notag \\ &+ \left(1-\frac{T^{\tr}}{T}\right) \min\left( \frac{ T^{\tr}}{T}\frac{\zeta}{\theta},1\right).  \label{equ:HiSEAB}
\end{align}

The results regarding sequential beamforming system without inter-node interference are first presented as an upper bound for the performance of proposed strategy. Then the results regarding the half-duplex systems are presented for comparison. Finally, the performance of sequential beamforming system with inter-node interference is presented.
 
%\textsl{Inter-node interference free sequential beamforming strategy}
\begin{theorem}[{Inter-node interference free sequential beamforming strategy}]\label{thm:HiABINIfree}
The maximal multiplexing gain of sequential beamforming strategy with closed loop training, without inter-node interference, under training power constraint $\zeta$ is
$$r_{\osc \cancel{\ini}}^{*}(\zeta)=\left\{\begin{matrix}
\frac{1}{2}\frac{\zeta}{\theta},\quad  \zeta <  \theta\\ 
1-\frac{1}{2}\frac{\theta}{\zeta},\quad \zeta  \geqslant  \theta
\end{matrix}\right. .$$
\end{theorem}

\begin{proof}
The multiplexing gain of sequential beamforming without inter-node interference is
$$ r_{\osc \cancel{\ini}}(\zeta,T^{\tr})=\left(1-\frac{1}{2}\frac{T^{\tr}}{T}\right)\min( \frac{\zeta}{\theta} \frac{T^{\tr}}{T},1). $$
By maximizing the multiplexing gain in the cases of $\frac{\zeta}{\theta} \frac{T^{\tr}}{T} \geqslant 1$ and $\frac{\zeta}{\theta} \frac{T^{\tr}}{T} <1$ by choosing the optimal training duration, the theorem is directly obtained.
\end{proof}

The multiplexing gain is composed of two regimes. When $\zeta$ is small, spectral efficiency increases linearly as training power increases. In this regime, the growth of rate performance during and after training is the primary reason. As more training power is allowed, users send training symbols until no spectral efficiency loss is observed due to inter-beam interference. The spectral efficiency improvement attributes to using less time to send the same amount of training information. Thus, spectral efficiency performance increases less as training power grows. Now the asymptotic performance of half-duplex counterpart is presented for comparison.
\begin{theorem}[Half-duplex system]\label{thm:HiHDC}
The maximal multiplexing gain of closed loop training half-duplex system under training power constraint $\zeta$ is
$$r_{\hdc}^{*}(\zeta)=\left\{\begin{matrix}
\frac{1}{4}\frac{\zeta}{\theta},\quad  \zeta < 2\theta\\ 
1-\frac{\theta}{\zeta},\quad \zeta  \geqslant 2 \theta
\end{matrix}\right. .$$
\end{theorem}

\begin{proof}
 Similar to inter-node interference free sequential beamforming strategy, omitting the extra spectral obtained during full-duplex training~\eqref{equ:HiSEAB}, we first express the multiplexing gain of half-duplex system as
$$ r_{\hdc}(\zeta,T^{\tr})=\left(1-\frac{T^{\tr}}{T}\right)\min( \frac{\zeta}{\theta} \frac{T^{\tr}}{T},1) .$$
Directly optimizing training duration in cases of $\frac{\zeta}{\theta} \frac{T^{\tr}}{T} \geqslant 1$ and $\frac{\zeta}{\theta} \frac{T^{\tr}}{T} <1$ leads to the proof.
\end{proof}

It should be noted that similar to sequential beamforming strategy with closed loop training, its half-duplex counterpart's spectral efficiency consists of two regimes. Compared to Theorem~\ref{thm:HiABINIfree}, a significant multiplexing gain improvement is observed. Thus, the proposed sequential beamforming strategy doubles the spectral efficiency of a unidirectional downlink communication asymptotically when $\zeta< \theta$. Finally, we look at the influence of inter-node interference on the asymptotic spectral efficiency of sequential beamforming strategy.

%\textsl{Sequential Bbeamforming strategy with inter-node interference}
\begin{theorem}[Sequential beamforming strategy with inter-node interference]\label{thm:HiABC}
The maximal multiplexing gain of closed loop training sequential beamforming strategy under training power constraint $\zeta$ is
\begin{equation*}
r_{\osc}^{*}\left(\zeta\right)=\left\{\begin{matrix}
r^{*}_{\osc \cancel{\ini}}\left(\zeta\right), \zeta\leqslant \frac{3\theta}{2+3\theta}\\ 
r^{*}_{\osc \ini}\left(\zeta\right), \frac{3\theta}{2+3\theta}<\zeta<  \min\left(1, \max\left(\frac{3\theta}{2+3\theta},\frac{3\theta}{2-\theta}\right)\right)\\
r^{*}_{\hdc}\left(\zeta\right), \min\left(1,\max\left(\frac{3\theta}{2+3\theta},\frac{3\theta}{2-\theta}\right)\right)\leqslant \zeta 
\end{matrix}\right. ,  
\end{equation*}
where 
\begin{equation*}
r_{\osc \ini}^{*}\left(\zeta\right)=\left\{\begin{matrix}
\frac{\left(\left(2-\theta\right)\zeta+\theta\right)^2 }{16\zeta\theta},\quad \zeta<  \frac{3\theta}{2-\theta}\\
1-\frac{\theta}{2}-\frac{\theta}{2\zeta},\quad \frac{3\theta}{2-\theta}\leqslant \zeta 
\end{matrix}\right. .   
\end{equation*}
\end{theorem}
\begin{proof}
Detailed proof can be found in Appendix~\ref{App:HiABC}.
\end{proof} 

The influence of inter-node interference on the spectral efficiency, interestingly, can be divided into three regimes. For systems targeting small multiplexing gain, only small amount of training power is needed. In this regime, inter-beam interference dominates the downlink performance during full-duplex training and no inter-node interference penalty is observed. However, if the higher multiplexing gain is targeted, the inter-node interference dominates the downlink performance during full-duplex training. In this case, inter-node interference will reduce the potential benefit obtained from sequential beamforming strategy. Finally, if a relatively high training power ($\zeta>1$) is used to achieve high spectral efficiency, the high inter-node interference level leads to no benefit from the downlink transmission during the training phase. This observation is confirmed by simulations shown in Fig.~\ref{fig:rVSzeta}.

\begin{remark}
As the number of antennas increases (with respect to block length), $\theta$ increases. Interestingly, we observe that higher $\theta$ actually increases the regime where sequential beamforming strategy does not suffer from inter-node interference. As $\theta \to \infty$, sequential beamforming strategy suffers no inter-node interference as long as training power constraint is smaller than 1.  
\end{remark}

\subsection{Sequential Beamforming with Open Loop Training} \label{sec:HighSNRO}
Following the same approach in Section~\ref{sec:HighSNRC}, we now investigate the spectral efficiency for different open loop training based systems. The influence of inter-beam interference in high $\snr$ regime is first characterized, which is followed by the multiplexing gain analysis.

\begin{Lemma}\label{lemma:IBI_HiSNRO}
For open loop training based systems, the rate performance under inter-beam interference after training is
\begin{align*}
\lim_{P \to \infty}\frac{ R^{\data}}{\log P}&=\lim_{P \to \infty}\frac{R^{\zf}-\log\left[1+\frac{P}{M}\frac{M-1}{1+\frac{T^{\tr}}{M}  P^{\zeta}}\right]}{\log P}\\&=\max\left(1-\zeta,0\right).
\end{align*}
\end{Lemma}
\begin{proof}
Substituting the training power to $P^{\zeta}$ in Proposition~\ref{pop:IBI} leads to the theorem. 
\end{proof} 

The rate performance achieved after training, surprisingly, is only decided by the training power constraint $\zeta$. More training symbols do not help to reduce inter-beam interference after training. Thus, the optimal training duration goes to zero in high $\snr$ regime. Therefore, the maximal multiplexing gain performance is obtained as follows.

\begin{theorem}\label{thm:IBI_HiSNRO}
For open loop training based systems, the multiplexing gain of both sequential beamforming and half-duplex strategy is only decided by training power constraint as
\begin{equation}
r_{\oso}^{*}=r_{\hdo}^{*}=\zeta.
\end{equation}
\end{theorem}
 
This theorem is valid for both half-duplex and sequential beamforming strategy with open loop training. Sequential beamforming strategy does not provide extra spectral benefit at high $\snr$. It should be emphasized that in low-to-moderate $\snr$ regime, from the analysis in Section~\ref{sec:optSE}, sequential beamforming does obtain significant spectral efficiency gain. This difference is not observed for closed loop training systems, where significant spectral efficiency improvement is observed in all $\snr$ regime. 
\begin{figure*}[htb]
\centering
\includegraphics[scale=1]{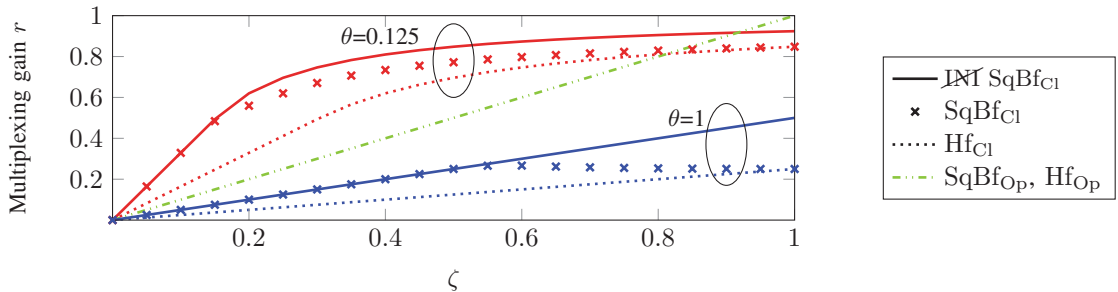} 
\caption{Multiplexing gain $r$ as a function of training power constraint $\zeta$.} 
\label{fig:rVSzeta}
\end{figure*} 

Before comparing the spectral efficiency performance of closed and open loop system asymptotically, we first validate the influence of the assumption that each user has a perfect knowledge of its channel on both closed and open loop training. This information is crucial for the decoding at the user side. It has been shown~\cite{caire2010multiuser} that, asymptotically, one training pilot from each base station antenna is both necessary and good enough for the influence of imperfect CSI on downlink rate to vanish for both closed and open loop training. 

Based on multiplexing gain characterization above, we compare the spectral efficiency of systems with different types of training. For systems with longer block length and few users, in general, closed loop training outperforms open loop training. The major reason is that the closed loop training significantly reduces inter-beam interference by learning from more training symbols. Despite the longer training duration, closed loop training is still more advantageous. However, if there are many antennas and block length is short, then it is better to use open loop training, whose training duration is asymptotically short. 
%\begin{remark} 
%Based on our analysis, we conclude the following about half-duplex systems. For systems with large numbers of antennas and short block length, which can be mathematically captured as $\theta \geqslant \frac{1}{2}$, open loop training outperforms closed loop training systems despite the choice of training power. Otherwise, training method should be picked based on the training power. When a strict training power constraint is imposed on users, closed loop training is more favorable by leveraging the training time. Evaluating Theorems~\ref{thm:HiHDC} and~\ref{thm:IBI_HiSNRO} gives the decision region as $\zeta<\frac{1-\sqrt{1-2\theta}}{2}$. However, if more training power is available, open loop training becomes more favorable due to shorter available training time. In general, training method should be chosen wisely to maximize the spectral efficiency.
%\end{remark}
\section{Conclusion}  \label{sec:conclude}
With accurate CSI, Multiuser MIMO downlink has the potential to increase the spectral efficiency tremendously. In systems with many users, CSI acquisition leads to unavoidable training overhead.  This paper aims at reducing the overhead of multiuser MIMO downlink systems by utilizing full-duplex radios. %Instead of requiring both base station and mobile users to be full-duplex capable~\cite{du2014mimo}~\cite{du2015SPAWC}, we propose a sequential beamforming strategy that requires only half-duplex users and less precoder updating.

With characterization of inter-node interference due to full-duplex training, we optimize the training duration of sequential beamforming strategy with both closed and open loop training. The proposed sequential beamforming strategy demonstrates significant spectral efficiency improvement compared to its half-duplex counterpart. Sequential beamforming strategy can also be applied in frequency-division duplex systems where uplink and downlink are orthogonal by nature. The orthogonality will also prevent the generation of inter-node interference.

The closed and open loop training methods exhibit distinct spectral efficiency performance. Asymptotically, the number of users, block length and training power jointly decide which type of training should be adopted. It has been observed in~\cite{marzetta2006fast} that closed loop training is more favorable than open loop training to reduce estimation error, while common sense suggests that open loop training is preferable for large systems. Our results quantify the decision region of training method and bridge these two observations.

We close this paper by noting some relevant issues. The training symbols are symmetrically allocated among users. Since the downlink receiving time decreases from User $1$ to User $M$,  extra spectral efficiency can be obtained by assigning training symbols decreasingly from User $1$ to User $M$. Moreover, we focus on providing the spectral efficiency characterization of a particular users subset to the scheduler. Thus, the MAC layer designs, such as deafness effect, are not considered. Past studies have revealed that picking users whose channel are more orthogonal reduces the inter-beam interference~\cite{yoo2007multi,dimic2005downlink}.  {  Finally, the considered base station is a full-duplex MIMO array where each antenna is used for both transmission and reception~\cite{bharadia2014full}. Another line of recent development of full-duplex designs~ \cite{duarte2012experiment,Jain2011practicalFD,everett2015measurement,aryafar2012midu,yin2013full} adopt separating transmission and reception antennas. For such array, the joint design and optimization of data service, CSI training, and full-duplex operation is an important future direction.}

%Finally, we consider channel realization to be independent and identically distributed Rayleigh fading, which could be viewed as a worst case and serves as performance lower bounded. Determining how to efficiently utilize the spatial information and the channel correlation between users is still an open question.
\begin{appendices} 
\section{Proof of Lemma~\ref{lem:INI}}\label{App:INI} 
{  Since perfect CSI is assumed available at each user, the rate loss can be upper bounded by
 \begin{align} 
\Delta R^{\tr}  \leqslant &\mathbb{E}\left[\log\left(1+|\mathbf{h}_{i}\mathbf{v}_i|^2\frac{P}{M}\right)\right] \notag \\&-\mathbb{E}\left[\log\left( 1+|\mathbf{h}_{i}\mathbf{v}_i|^2\frac{P}{M} +|\mathbf{h}_{ik}x_{\mathrm{tr}_{k}}|^2\right)\right] 
\notag\\&+\mathbb{E}\left[\log\left( 1+\sum_{j\neq i}\frac{P}{M}|\mathbf{h}_{i}\mathbf{v}_j|^2+|\mathbf{h}_{ik}x_{\mathrm{tr}_{k}}|^2\right)\right].\label{equ:PrDeR}
\end{align}
%\begin{align} 
%\Delta R^{\tr} \stackrel{a)}{\leqslant} & \mathbb{E}\left[\log\left(1+|\mathbf{h}_{i}\mathbf{V}_i|^2\frac{P}{M}\right)\right]
%-\mathbb{E}\left[\log\left(1+\frac{|\mathbf{h}_{i}\mathbf{V}_i|^2 \frac{P}{M}}{1+\sum_{j\neq i}\frac{P}{M}|\mathbf{h}_{i}\mathbf{V}_j|^2+|\mathbf{h}_{ik}x_{\mathrm{tr}_{k}}|^2}\right)\right] \notag
%\notag\\ \leqslant &\mathbb{E}\left[\log\left(1+|\mathbf{h}_{i}\mathbf{V}_i|^2\frac{P}{M}\right)\right] -\mathbb{E}\left[\log\left( 1+|\mathbf{h}_{i}\mathbf{V}_i|^2\frac{P}{M} +|\mathbf{h}_{ik}x_{\mathrm{tr}_{k}}|^2\right)\right]\notag
%\notag\\&+\mathbb{E}\left[\log\left( 1+\sum_{j\neq i}\frac{P}{M}|\mathbf{h}_{i}\mathbf{V}_j|^2+|\mathbf{h}_{ik}x_{\mathrm{tr}_{k}}|^2\right)\right].\label{equ:PrDeR}
%\end{align}
The first step is to follow the same recipe in Appendix II step (a) of~\cite{caire2010multiuser}, and then ignoring the positive term $\sum_{j\neq i}\frac{P}{M}|\mathbf{h}_{i}\mathbf{v}_j|^2$ leads to the result above. The sum $\sum_{j\neq i}\frac{P}{M}|\mathbf{h}_{i}\mathbf{v}_j|^2$ and the term $|\mathbf{h}_{ik}x_{\mathrm{tr}_{k}}|^2$ quantify the interference due to imperfect precoding and full-duplex training, i.e., inter-beam and inter-node interference, respectively. We refer them as $\mathcal{P}_{\mathrm{IBI}}$ and $\mathcal{P}_{\ini}$. {  The inter-node interference is assumed to scale proportional to the training power of users, i.e., $\mathcal{P}_{\ini}=\alpha f P$.} Noting the concavity of logarithm function, similar to step (33) in Remark 4.2 in~\cite{caire2010multiuser}, we apply Jensen's inequality to obtain Lemma~\ref{lem:INI}. } 
%to Eq.~\eqref{equ:PrDeR} to obtain
%\begin{equation*}
%\Delta R^{\tr}\leqslant \log\left( \frac{1+ \mathcal{P}_{\mathrm{IBI}}+\alpha f P}{1 +\frac{\alpha f P}{1+\frac{P}{M}}}\right).
%\end{equation*} 
%
%Directly using the results from Section V and Remark 4.2 in~\cite{caire2010multiuser}, the influence of inter-beam interference can be captured as
%\begin{equation*}
% \mathcal{P}_{\mathrm{IBI}}=\left\{\begin{matrix}
%P(1+fP)^{-\frac{T^{\tr}_{i}}{M\left(M-1\right)} } ,\quad  \text{for closed loop training}\\
%\frac{P}{M}\frac{M-1}{1+\frac{T^{\tr}_{i}}{M} f P},\quad  \text{for open loop training}
%\end{matrix}\right. , 
%\end{equation*}
%which leads to the theorem. 
\section{Proof of Theorem~\ref{thm:LenABC}}\label{App:OptLenABC} 
The rate increase term $\delta R^{\data}\left(T^{\tr}\right)$ of~\eqref{equ:optLenABfin} is characterized by applying Proposition~\ref{pop:IBI} as
%=&  \log\left(\frac{1+P\left(1+fP\right)^{-\frac{T^{\tr}}{M\left(M-1\right)}}}{1+P\left(1+fP\right)^{-\frac{T^{\tr}}{M\left(M-1\right)}-\frac{1}{M-1}}}\right)\\ 
%= \log\left(1+P\left(1+fP\right)^{-VSze\frac{T^{\tr}}{M\left(M-1\right)}} \right)-\log\left(1+P\left(1+fP\right)^{-\frac{T^{\tr}}{M\left(M-1\right)}-\frac{1}{M-1}} \right)\\
$$
\delta R^{\data}\left(T^{\tr}\right)
\approx P\left(1+fP\right)^{-\frac{T^{\tr}}{M\left(M-1\right)}}[\log\left(1+fP\right)^{\frac{1}{M-1}}-1].
$$
Here the last step is directly obtained by using Taylor expansion. Substituting  into~\eqref{equ:optLenABfin}, we have
  \begin{align*}
&\left(1-\frac{M+1}{2M}\frac{T^{\tr}}{T}\right)P\left(1+fP\right)^{-\frac{T^{\tr}}{M\left(M-1\right)}}[\left(1+fP\right)^{\frac{1}{M-1}}-1]\\
&=\frac{M-1}{2T}  \Delta R^{\ini}+\frac{M+1}{2T} R^{\zf}.   
  \end{align*}
Noticing that it is an transcendental equation which is challenging to solve. Then omitting the $\frac{M+1}{2M}\frac{T^{\tr}}{T}$ term leads us to the theorem. This approximation is valid for large $T$.

\section{Proof of Theorem~\ref{thm:LenABO}}\label{App:OptLenABO}
For open loop training based systems, the rate improvement due to more training symbols can be obtained by using Proposition~\ref{pop:IBI} as
%& \log\left(1+\frac{P}{ M}\frac{M-1}{1+f  \widetilde{T}_{\os}^{\tr} P/ M } \right)-\log\left(1+\frac{P}{ M}\frac{M-1}{1+f  \widetilde{T}_{\os}^{\tr} P/ M +fP} \right)\\
\begin{equation*}
\delta R^{\data}\left(T^{\tr}\right)
\approx \frac{(M-1)M} { f\left(\widetilde{T}_{\os}^{\tr}\right)^2   },
\end{equation*}
where the last step is the direct result of Maclaurin series. Combining with~\eqref{equ:optLenABfin} leads to
  \begin{equation*}
\left(1-\frac{M+1}{2M}\frac{\widetilde{T}^{\tr}_{\os}}{T}\right) \frac{(M-1)M} { f\left(\widetilde{T}_{\os}^{\tr}\right)^2   }=\frac{M-1}{2T}  \Delta R^{\ini}+\frac{M+1}{2T} R^{\zf},
\end{equation*} 
whose solution leads to the theorem. 
\section{Proof of Theorem~\ref{thm:LenHDC}}\label{App:OptLenHDC}
The rate increase term $\delta R^{\data}\left(T^{\tr}\right)$ is immediately approximated by using results from Appendix~\ref{App:OptLenABC}. Applying this rate characterization term to evaluate ~\eqref{equ:optHDMU} gives
$$
\left(1-\frac{T^{\tr}}{T}\right)\frac{\left(1+fP\right)^{\frac{1}{M-1}}-1}{P\left(1+fP\right)^{-\frac{T^{\tr}}{M\left(M-1\right)}}}=\frac{M}{T}R^{\zf},
$$
which is a transcendental equation. Following the same step in Appendix~\ref{App:OptLenABC}, we omit the $\frac{T^{\tr}}{T}$ to obtain the theorem. 
\section{Proof of Theorem~\ref{thm:LenHDO}}\label{App:OptLenHDO}
 Using the rate increase characterization term $\delta R^{\data}\left(T^{\tr}\right)$ in Appendix~\ref{App:OptLenABO} and further applying ~\eqref{equ:optHDMU},~\eqref{equ:optHDMC} gives
  \begin{equation*}
\left(1- \frac{T^{\tr}}{T}\right)\frac{(M-1)M} { f\left(\widetilde{T}_{\os}^{\tr}\right)^2   }=\frac{M}{T}R^{\zf},
\end{equation*} 
whose solution is the theorem.
% \begin{align*}
%\delta R^{\data}\left(T^{\tr}\right)\approx \frac{(M-1)M} { f\left(\widetilde{T}_{\os}^{\tr}\right)^2   }.
%\end{align*} 
%%%%%%%%%%%%%%%%%%%%%%%%%%%%%%%Spectral efficiency Proof%%%%%%%%%%%%%%%%%%%%%%%%%%%%%%%%%%%%%%%%%%%%
\section{Proof of Theorem~\ref{thm:SEABC}}\label{App:OptSEABC}
Evaluating the achieved spectral efficiency of sequential beamforming strategy with approximated optimal training duration $\widetilde{T}_{\osc}^{\tr *}$ obtained in Theorem~\ref{thm:LenABC} gives upper bound
%\leqslant& \se^{\zf}-\se_{\osc}\left(\widetilde{T}_{\osc}^{\tr *}\right)\\
 \begin{align*}
&\Delta \se_{\osc}\left(T_{\osc}^{\tr *}\right)
 \leqslant  R^{\zf}+\frac{M-1}{2M}\frac{\widetilde{T}_{\osc}^{\tr *}}{T} \Delta R^{\ini} \\
 &-\left(1-\frac{M+1}{2M}\frac{\widetilde{T}_{\osc}^{\tr *}}{T}\right)\left[R^{\zf}-\log\left(  1+P\left(1+fP\right)^{-\frac{\widetilde{T}_{\osc}^{\tr *}}{M(M-1)}}\right)\right].
\end{align*}
Omitting the negative term $$-\frac{M+1}{2M}\frac{\widetilde{T}_{\osc}^{\tr *}}{T}\log\left(  1+P\left(1+fP\right)^{-\frac{\widetilde{T}_{\osc}^{\tr *}}{M(M-1)}}\right)$$ and sorting the small term with respect to $\frac{1}{\sqrt{T}}$ lead to the theorem.
%\begin{equation*}
%\Delta \se_{\osc}\left(T_{\osc}^{\tr *}\right)\leqslant   \frac{M-1}{2M}\frac{\widetilde{T}_{\osc}^{\tr *}}{T} \Delta R^{\ini}+\frac{M+1}{2M}\frac{\widetilde{T}_{\osc}^{\tr *}}{T} R^{\zf}+\log\left(  1+P\left(1+fP\right)^{-\frac{\widetilde{T}_{\osc}^{\tr *}}{M(M-1)}}\right).
%\end{equation*}
%The theorem is obtained by observing the last term to be a small term with respect to $\frac{1}{\sqrt{T}}$. 
\section{Proof of Theorem~\ref{thm:SEABO}}\label{App:OptSEABO}
Following similar analysis as that of Appendix~\ref{App:OptSEABC}, we substitute approximation of optimal training duration from Theorem~\ref{thm:LenABO} into~\eqref{equ:SEosDeatiled} to characterize the spectral efficiency loss as
%\leqslant& \se^{\zf}-\se_{\oso}\left(\widetilde{T}_{\oso}^{\tr *}\right)\\
\begin{align*}
 &\Delta \se_{\oso}\left(T_{\oso}^{\tr *}\right)
 \leqslant  R^{\zf}+\frac{M-1}{2M}\frac{\widetilde{T}_{\oso}^{\tr *}}{T} \Delta R^{\ini} \\
 &-\left(1-\frac{M+1}{2M}\frac{\widetilde{T}_{\oso}^{\tr *}}{T}\right)\left[R^{\zf}-\log\left(  1+\frac{(M-1)\frac{P}{M}}{1+ f  \widetilde{T}_{\oso}^{*}P/M}\right)\right].
 \end{align*}
 Then dropping the negative term $-\frac{M+1}{2M}\frac{\widetilde{T}_{\oso}^{\tr *}}{T}\log\left(  1+\frac{(M-1)\frac{P}{M}}{1+ f  \widetilde{T}_{\oso}^{*}P/M}\right)$  and sorting small term with respect to $\frac{\log T}{T}$ lead to the theorem.
% \begin{equation*}
% \Delta \se_{\oso}\left(T_{\oso}^{\tr *}\right) \leqslant  \frac{M-1}{2M}\frac{\widetilde{T}_{\oso}^{\tr *}}{T} \Delta R^{\ini}+\frac{M+1}{2M}\frac{\widetilde{T}_{\oso}^{\tr *}}{T} R^{\zf}+\log\left(  1+\frac{(M-1)\frac{P}{M}}{1+ f  \widetilde{T}_{\oso}^{*}P/M}\right).
% \end{equation*}
%Theorem is proved by noticing the last term is a small term with respect to $\frac{\log T}{T}$.

% =&  R^{\zf}-\frac{M-1}{2M}\frac{\widetilde{T}_{\oso}^{\tr *}}{T}\left[R^{\zf}- \log\left( \frac{1+\frac{(M-1)\frac{P}{M}}{1+ f  \widetilde{T}_{\oso}^{*}P/M}+\alpha f P}{1 +\frac{\alpha f P}{1+\frac{P}{M}}}\right)\right]\\ &-\left(1-\frac{\widetilde{T}_{\oso}^{\tr *}}{T}\right)\left[R^{\zf}-\log\left(  1+\frac{(M-1)\frac{P}{M}}{1+ f  \widetilde{T}_{\oso}^{*}P/M}\right)\right] \\
\section{Proof of Theorem~\ref{thm:SEHDC}}\label{App:OptSEHDC} 
Similar to systems adopting sequential beamforming strategy, the spectral efficiency gap of the half-duplex counterparts with respect to the genie-aided scenario can be immediately upper bounded by evaluating the sub-optimal scheme $\se_{\hdc}\left(\widetilde{T}_{\hdc}^{*}\right)$: 
%\\
%\leqslant \se^{\zf}-\se_{\hdc}\left(\widetilde{T}_{\hdc}^{\tr *}\right)
 \begin{align*}
\Delta \se_{\hdc}\left(T_{\hdc}^{\tr *}\right)
\stackrel{(a)}{\leqslant} & \frac{T_{\hdc}^{\tr *}}{T} R^{\zf}+\log\left(  1+P\left(1+fP\right)^{-\frac{\widetilde{T}_{\hdc}^{\tr *}}{M(M-1)}}\right)
\\ =& R^{\zf} \frac{M\left(M-1\right)}{\log\left(1+fP\right)}\frac{\log T}{T}+o\left(\frac{\log T}{T}\right).
\end{align*}
% =&  R^{\zf}-\left(\frac{T-\widetilde{T}_{\hdc}^{\tr *}}{T}\right)\left[R^{\zf}-\log\left(  1+P\left(1+fP\right)^{-\frac{\widetilde{T}_{\hdc}^{\tr *}}{M(M-1)}}\right)\right] \\
Inequality $(a)$ is the result of dropping negative term $-\frac{\widetilde{T}_{\hdc}^{*}}{T}\log\left(  1+P\left(1+fP\right)^{-\frac{\widetilde{T}_{\hdc}^{\tr *}}{M(M-1)}}\right)$. Applying training time approximation in Theorem~\ref{thm:LenHDC} gives the final step.
  
\section{Proof of Theorem~\ref{thm:SEHDO}}\label{App:OptSEHDO}
Inspired by~\cite{kobayashi2011training}, spectral efficiency gap with respect to genie-aided situation can be immediately upper bounded by evaluating $\se_{\hdo}\left(\widetilde{T}_{\hdo}^{*}\right)$
 \begin{align*}
\Delta \se_{\hdo}\left(T_{\hdo}^*\right)\leqslant& \se^{\zf}-\se_{\hd}\left(\widetilde{T}_{\hdo}^{*}\right)\\ \stackrel{(a)}{\leqslant}&  \frac{\widetilde{T}_{\hd}^{*}}{T}R^{\zf}+\log\left(  1+\frac{(M-1)\frac{P}{M}}{1+ f  \widetilde{T}_{\hdo}^{*}P/M}\right) \\
\leqslant& \sqrt{\frac{(M-1)R^{\zf}}{fT}}+\sqrt{\frac{R^{\zf}(M-1)}{fT}}\\=& 2\sqrt{\frac{(M-1)R^{\zf}}{fT}}.
\end{align*}
% =& R^{\zf}-\left(\frac{T-T_{\hd}^*}{T}\right)\left[R^{\zf}-\log\left(  1+\frac{(M-1)\frac{P}{M}}{1+ f  \widetilde{T}_{\hdo}^{*}P/M}\right)\right]\\
Inequality $(a)$ is obtained by dropping negative term $-\frac{\widetilde{T}_{\hdo}^{*}}{T}\log\left(  1+\frac{(M-1)\frac{P}{M}}{1+ f  \widetilde{T}_{\hd}^{*}P/M}\right)$. The next step is the result of
Maclaurin expansion of the logarithm term, which is tight for large $T$. 

\section{Proof of Theorem~\ref{thm:HiABC}}\label{App:HiABC}
Studying~\eqref{equ:HiSEAB} in different regimes of operation gives the following.
\begin{enumerate}
\item When $\zeta \geqslant 1$, the multiplexing gain is
%$$r_{\osc}\left(\zeta,T^{tr}\right)= \left(1-\frac{T^{\tr}}{T}\right) \frac{ T^{\tr}}{T} \frac{\zeta}{\theta}.$$
%Notice that this is
the same as the multiplex gain of half-duplex systems. Thus, $ r^{*}_{\osc}(\zeta)=r^{*}\left(\zeta\right)_{\hdc}.  $

\item When $\zeta < 1$

\begin{itemize}
\item  $\frac{ T^{\tr}}{T} \leqslant  \theta\frac{1-\zeta}{\zeta}$:
\begin{align*}
r_{\osc}\left(\zeta,T^{tr}\right)&=   \frac{ T^{\tr}}{T}\frac{\zeta}{\theta} -\frac{1}{2}\left(\frac{T^{\tr}}{T}\right)^2 \frac{\zeta}{\theta}\\
&=r_{\osc \cancel{\ini}}\left(\zeta,T^{tr}\right)
\end{align*}
 
\item  $\theta\frac{1-\zeta}{\zeta}<\frac{ T^{\tr}}{T} <\frac{\theta}{\zeta} $: 
$$
r_{\osc}\left(\zeta,T^{tr}\right)=\frac{1}{2}\frac{T^{\tr}}{T}  \left( 1-\zeta\right)+ \left(1-\frac{T^{\tr}}{T}\right)  \frac{ T^{\tr}}{T}\frac{\zeta}{\theta} 
$$
 
\item  $\frac{\theta}{\zeta} \leqslant \frac{ T^{\tr}}{T}$:
$$
r_{\osc}\left(\zeta,T^{tr}\right)=  \frac{1}{2} \frac{ T^{\tr}}{T}\left(1-\zeta\right) +\left( 1-\frac{T^{\tr}}{T}\right) 
$$
\end{itemize}
By carefully evaluating the derivative in different regimes and applying the optimized training duration into equation~\eqref{equ:HiSEAB} lead to the theorem.
%, for $\zeta<1$,
%\begin{equation*}
%\frac{ T^{\tr *}}{T}=\left\{\begin{matrix}
%1,\quad \zeta\leqslant \frac{3\theta}{2+3\theta}\\ 
%\frac{1}{2}+\frac{\theta\left(1-\zeta\right)}{4\zeta},\quad \frac{3\theta}{2+3\theta}<\zeta<  \frac{3\theta}{2-\theta}\\
%\frac{\theta}{\zeta},\quad \frac{3\theta}{2-\theta}\leqslant \zeta 
%\end{matrix}\right. .   
%\end{equation*}
%\begin{equation*}
%r^{*}\left(\zeta\right)=\left\{\begin{matrix}
%r^{*}_{\osc \cancel{\ini}}\left(\zeta\right),\quad \zeta\leqslant \frac{3\theta}{2+3\theta}\\ 
%\frac{\left(\left(2-\theta\right)\zeta+\theta\right)^2 }{16\zeta\theta},\quad \frac{3\theta}{2+3\theta}<\zeta<  \frac{3\theta}{2-\theta}\\
%1-\frac{\theta}{2}-\frac{\theta}{2\zeta},\quad \frac{3\theta}{2-\theta}\leqslant \zeta 
%\end{matrix}\right. .   
%\end{equation*}
\end{enumerate}   
\end{appendices}
\section*{Acknowledgment}
\addcontentsline{toc}{section}{Acknowledgment}
We want to thank the authors of~\cite{everett2015measurement} for providing the measurement data used in Section~\ref{sec:optSE}.
\bibliographystyle{ieeetr}
\bibliography{SqBf_ref}
\end{document}